\pdfoutput=1

\documentclass[%
showpacs,%
showkeys,%
superscriptaddress,%
groupedaddress,%
longbibliography,%
aps,%
pra,%
twocolumn%
]{revtex4-2}

\usepackage[T1]{fontenc}
\usepackage[utf8]{inputenc}
\usepackage[english]{babel}

\usepackage{lipsum}
\usepackage{comment}
\usepackage{centernot}

\usepackage{graphicx}
\usepackage{epsfig}
\usepackage[all,cmtip,matrix,frame,arrow]{xy}
\usepackage{tabularx}
\usepackage{booktabs}
\usepackage[usenames,dvipsnames,svgnames,table,xcdraw]{xcolor}
\usepackage{soul}
\usepackage{bigstrut}

\usepackage{amsfonts}
\usepackage{amsthm}
\usepackage{amsmath}
\usepackage{amssymb}
\usepackage{mathrsfs}
\usepackage{faktor}
\usepackage{stmaryrd}
\usepackage{mathtools}

\usepackage[normalem]{ulem}

\usepackage{etoolbox}
\usepackage{enumerate}
\makeatletter
\newcommand\setItemnumber[1]{\setcounter{enum\romannumeral\@enumdepth}{\numexpr#1-1\relax}}
\makeatother

\usepackage{tikz}
\input{tikzit.sty}
\tikzstyle{medium rectangle}=[fill=white, draw=black, shape=rectangle, minimum width=0.75 cm, minimum height=1 cm]
\tikzstyle{red}=[fill=red, draw=black, shape=circle]
\tikzstyle{system_label}=[fill=none, draw=none, shape=circle]
\tikzstyle{small square}=[fill=white, draw=black, shape=rectangle]
\tikzstyle{big rectangle}=[fill=white, draw=black, shape=rectangle, minimum width=1.2 cm, minimum height=1.5 cm]
\tikzstyle{tall rectangle}=[fill=white, draw=black, shape=rectangle, minimum width=1.2 cm, minimum height=2.1 cm]
\tikzstyle{medium square}=[fill=white, draw=black, shape=rectangle, minimum width=0.6 cm, minimum height=0.6 cm]
\tikzstyle{huge rectangle}=[fill=white, draw=black, shape=rectangle, minimum width=2 cm, minimum height=4.5 cm]
\tikzstyle{yellow small}=[fill={rgb,255: red,255; green,252; blue,144}, draw=black, shape=rectangle]
\tikzstyle{violet node}=[fill={rgb,255: red,195; green,187; blue,255}, draw=black, shape=rectangle, minimum width=1.2 cm, minimum height=1.5 cm]
\tikzstyle{pink node}=[fill={rgb,255: red,238; green,174; blue,255}, draw=black, shape=rectangle, minimum width=1.2 cm, minimum height=1.5 cm]
\tikzstyle{pompelmo node}=[fill={rgb,255: red,255; green,187; blue,166}, draw=black, shape=rectangle, minimum width=1.2 cm, minimum height=1.5 cm]
\tikzstyle{ottano node}=[fill={rgb,255: red,131; green,201; blue,187}, draw=black, shape=rectangle]
\tikzstyle{blue}=[fill={rgb,255: red,214; green,201; blue,255}, draw=black, shape=rectangle]
\tikzstyle{pure}=[fill=blue, draw=blue, shape=circle, minimum size=8pt, inner sep=0pt, outer sep=0pt]
\tikzstyle{extremal}=[fill={rgb,255: red,0; green,233; blue,0}, draw={rgb,255: red,0; green,233; blue,0}, shape=circle, minimum size=8pt, inner sep=0pt, outer sep=0pt]
\tikzstyle{atomic - extremal}=[fill={rgb,255: red,0; green,167; blue,0}, draw={rgb,255: red,0; green,167; blue,0}, shape=circle, minimum size=8pt, inner sep=0pt, outer sep=0pt]

\tikzstyle{lambda}=[-, draw={rgb,255: red,191; green,191; blue,191}]
\tikzstyle{state}=[<-]
\tikzstyle{bluette}=[-, fill={rgb,255: red,214; green,201; blue,255}, draw={rgb,255: red,192; green,181; blue,229}]
\tikzstyle{greenish}=[-, fill={rgb,255: red,160; green,217; blue,255}, draw={rgb,255: red,137; green,188; blue,219}]
\tikzstyle{white}=[-, fill=white]
\tikzstyle{reddish}=[-, fill={rgb,255: red,255; green,143; blue,145}, dashed]
\tikzstyle{yellowish}=[-, fill={rgb,255: red,255; green,252; blue,144}, dashed, line width=1 pt]
\tikzstyle{red non-dashed}=[-, fill={rgb,255: red,255; green,143; blue,145}]
\tikzstyle{violet}=[-, dashed, fill={rgb,255: red,195; green,187; blue,255}]
\tikzstyle{green non-dashed}=[-, fill={rgb,255: red,160; green,217; blue,255}]
\tikzstyle{pompelmo non-dashed}=[-, fill={rgb,255: red,255; green,187; blue,166}]
\tikzstyle{pink}=[-, fill={rgb,255: red,238; green,174; blue,255}, dashed]
\tikzstyle{blue non-dashed}=[-, fill={rgb,255: red,178; green,255; blue,246}]
\tikzstyle{yell non-dashed}=[-, fill={rgb,255: red,255; green,248; blue,137}]
\tikzstyle{ottano}=[-, fill={rgb,255: red,164; green,184; blue,255}, draw={rgb,255: red,147; green,168; blue,229}]
\tikzstyle{pompelmo}=[-, fill={rgb,255: red,255; green,187; blue,166}, draw={rgb,255: red,213; green,156; blue,139}]
\tikzstyle{dark blue}=[-, fill={rgb,255: red,131; green,201; blue,187}, draw={rgb,255: red,119; green,183; blue,170}]
\tikzstyle{d.blue non-dashed}=[-, fill={rgb,255: red,237; green,148; blue,112}]
\tikzstyle{dashed edge}=[-, dashed]
\tikzstyle{ottano non-drawn}=[-, fill={rgb,255: red,164; green,184; blue,255}]
\tikzstyle{atomic}=[-, fill={rgb,255: red,232; green,232; blue,232}, draw={rgb,255: red,255; green,128; blue,0}, line width=1.5pt]
\tikzstyle{atomic - dashed}=[-, dashed, fill={rgb,255: red,232; green,232; blue,232}, draw={rgb,255: red,255; green,128; blue,0}, line width=0.2pt]
\tikzstyle{det}=[-, fill={rgb,255: red,232; green,232; blue,232}, draw={rgb,255: red,227; green,0; blue,3}, line width=1.5pt]
\tikzstyle{det - dashed}=[-, dashed, fill={rgb,255: red,232; green,232; blue,232}, draw={rgb,255: red,227; green,0; blue,3}, line width=0.2pt]
\tikzstyle{standard}=[-, fill={rgb,255: red,232; green,232; blue,232}, line width=1.5pt]
\tikzstyle{empty}=[-, fill={rgb,255: red,232; green,232; blue,232}, draw={rgb,255: red,232; green,232; blue,232}, line width=1.5pt]
\tikzstyle{empty- dashed}=[-, dashed, fill={rgb,255: red,232; green,232; blue,232}, draw={rgb,255: red,232; green,232; blue,232}, line width=1.5pt]
\tikzstyle{arrow}=[draw={rgb,255: red,255; green,128; blue,0}, ->, line width=2pt]

\makeatletter
\newcommand*{\hyperlinkcite}[1]{\hyper@link{cite}{cite.#1}}%\hyperlinkcite takes 2 arguments: #1<- cite-key, #2<- link-text
\makeatother

\usepackage[]{hyperref}
\hypersetup{colorlinks, citecolor=blue, urlcolor=blue, linkcolor=blue,linktocpage,breaklinks}
\usepackage[nameinlink,noabbrev]{cleveref}

\makeatletter
\pretocmd{\NAT@citexnum}{\@ifnum{\NAT@ctype>\z@}{\let\NAT@hyper@\relax}{}}{}{}
\makeatother

\newcommand{\qw}[1][-1]{\ar @{-} [0,#1]}
\newcommand{\gate}[1]{*{\xy *+<.6em>{#1};p\save+LU;+RU **\dir{-}\restore\save+RU;+RD **\dir{-}\restore\save+RD;+LD **\dir{-}\restore\POS+LD;+LU **\dir{-}\endxy} \qw}
\newcommand{\measureD}[1]{*{\xy*+=+<.5em>{\vphantom{\rule{0em}{.1em}#1}}*\cir{r_l};p\save*!R{#1} \restore\save+UC;+UC-<.5em,0em>*!R{\hphantom{#1}}+L **\dir{-} \restore\save+DC;+DC-<.5em,0em>*!R{\hphantom{#1}}+L **\dir{-} \restore\POS+UC-<.5em,0em>*!R{\hphantom{#1}}+L;+DC-<.5em,0em>*!R{\hphantom{#1}}+L **\dir{-} \endxy} \qw}
\newcommand{\multimeasureD}[2]{*+<1em,.9em>{\hphantom{#2}}\save[0,0].[#1,0];p\save !C *{#2},p+LU+<0em,0em>;+RU+<-.8em,0em> **\dir{-}\restore\save +LD;+LU **\dir{-}\restore\save +LD;+RD-<.8em,0em> **\dir{-} \restore\save +RD+<0em,.8em>;+RU-<0em,.8em> **\dir{-} \restore \POS !UR*!UR{\cir<.9em>{r_d}};!DR*!DR{\cir<.9em>{d_l}}\restore \qw}
\newcommand{\multigate}[2]{*+<1em,.9em>{\hphantom{#2}} \qw \POS[0,0].[#1,0];p !C *{#2},p \save+LU;+RU **\dir{-}\restore\save+RU;+RD **\dir{-}\restore\save+RD;+LD **\dir{-}\restore\save+LD;+LU **\dir{-}\restore}
\newcommand{\ghost}[1]{*+<1em,.9em>{\hphantom{#1}} \qw}
\newcommand{\ustick}[1]{*!D!<0em,-.5em>=<0em>{#1}}
\newcommand{\Qcircuit}[1][0em]{\xymatrix @*=<#1>}
\newcommand{\node}[2][]{{\begin{array}{c} \ _{#1}\  \\ {#2} \\ \ \end{array}}\drop\frm{o} }
\newcommand{\pureghost}[1]{*+<1em,.9em>{\hphantom{#1}}}
\newcommand{\multiprepareC}[2]{*+<1em,.9em>{\hphantom{#2}}\save[0,0].[#1,0];p\save !C
  *{#2},p+RU+<0em,0em>;+LU+<+.8em,0em> **\dir{-}\restore\save +RD;+RU **\dir{-}\restore\save
  +RD;+LD+<.8em,0em> **\dir{-} \restore\save +LD+<0em,.8em>;+LU-<0em,.8em> **\dir{-} \restore \POS
  !UL*!UL{\cir<.9em>{u_r}};!DL*!DL{\cir<.9em>{l_u}}\restore}
\newcommand{\prepareC}[1]{*{\xy*+=+<.5em>{\vphantom{#1\rule{0em}{.1em}}}*\cir{l^r};p\save*!L{#1} \restore\save+UC;+UC+<.5em,0em>*!L{\hphantom{#1}}+R **\dir{-} \restore\save+DC;+DC+<.5em,0em>*!L{\hphantom{#1}}+R **\dir{-} \restore\POS+UC+<.5em,0em>*!L{\hphantom{#1}}+R;+DC+<.5em,0em>*!L{\hphantom{#1}}+R **\dir{-} \endxy}}

\newcommand{\myQcircuit}[1]{\begin{aligned} \Qcircuit @C=0.8em @R=0.8em {#1} \end{aligned} }

\theoremstyle{definition}
\newtheorem*{definition*}{Definition}
\newtheorem{definition}{Definition}

\newcommand{\textdef}[1]{\textit{{#1}}}

\theoremstyle{plain}

\newtheorem*{property*}{Property}

\newtheorem{theorem}{Theorem}
\newtheorem*{theorem*}{Theorem}
\newtheorem{corollary}{Corollary}
\newtheorem*{corollary*}{Corollary}

\newtheorem*{proposition*}{Proposition}

\newtheorem*{conjecture*}{Conjecture}
\newtheorem*{question*}{Question}

\newtheorem*{problem*}{Problem}

\newtheorem*{lemma*}{Lemma}
\newtheorem{lemma}{Lemma}

\newtheorem*{example*}{Example}

\newtheorem*{remark*}{Remark}

\newtheorem{remark}{Remark}
\newtheorem*{proof*}{Proof}

\makeatletter
\newtheorem*{rep@theorem}{\rep@title}
\newcommand{\newreptheorem}[2]{%
	\newenvironment{rep#1}[1]{%
		\def\rep@title{\autoref{##1}}%
		\begin{rep@theorem}}%
		{\end{rep@theorem}}}
\makeatother
\newreptheorem{theorem}{Theorem}

\DeclareSymbolFont{sfletters}{OML}{cmbrm}{m}{it}
\DeclareMathSymbol{\salpha}{\mathord}{sfletters}{"0B}
\DeclareMathSymbol{\sbeta}{\mathord}{sfletters}{"0C}
\DeclareMathSymbol{\sgamma}{\mathord}{sfletters}{"0D}
\DeclareMathSymbol{\sdelta}{\mathord}{sfletters}{"0E}
\DeclareMathSymbol{\sepsilon}{\mathord}{sfletters}{"0F}
\DeclareMathSymbol{\szeta}{\mathord}{sfletters}{"10}
\DeclareMathSymbol{\seta}{\mathord}{sfletters}{"11}
\DeclareMathSymbol{\stheta}{\mathord}{sfletters}{"12}
\DeclareMathSymbol{\siota}{\mathord}{sfletters}{"13}
\DeclareMathSymbol{\skappa}{\mathord}{sfletters}{"14}
\DeclareMathSymbol{\slambda}{\mathord}{sfletters}{"15}
\DeclareMathSymbol{\smu}{\mathord}{sfletters}{"16}
\DeclareMathSymbol{\snu}{\mathord}{sfletters}{"17}
\DeclareMathSymbol{\sxi}{\mathord}{sfletters}{"18}
\DeclareMathSymbol{\spi}{\mathord}{sfletters}{"19}
\DeclareMathSymbol{\srho}{\mathord}{sfletters}{"1A}
\DeclareMathSymbol{\ssigma}{\mathord}{sfletters}{"1B}
\DeclareMathSymbol{\stau}{\mathord}{sfletters}{"1C}
\DeclareMathSymbol{\supsilon}{\mathord}{sfletters}{"1D}
\DeclareMathSymbol{\sphi}{\mathord}{sfletters}{"1E}
\DeclareMathSymbol{\schi}{\mathord}{sfletters}{"1F}
\DeclareMathSymbol{\spsi}{\mathord}{sfletters}{"20}
\DeclareMathSymbol{\somega}{\mathord}{sfletters}{"21}
\DeclareMathSymbol{\svarepsilon}{\mathord}{sfletters}{"22}
\DeclareMathSymbol{\svartheta}{\mathord}{sfletters}{"23}
\DeclareMathSymbol{\svarpi}{\mathord}{sfletters}{"24}
\DeclareMathSymbol{\svarrho}{\mathord}{sfletters}{"25}
\DeclareMathSymbol{\svarsigma}{\mathord}{sfletters}{"26}
\DeclareMathSymbol{\svarphi}{\mathord}{sfletters}{"27}

\def\<{\langle}\def\>{\rangle}

\newcommand{\mathDef}{\coloneq}

\newcommand{\kronekerDelta}[2]{\delta_{ #1 , #2}}

\newcommand*{\hilbert}{\mathcal{H}}

\newcommand*{\cardinality}[1]{\ensuremath{\vert #1 \vert}}

\newcommand{\cartesianC}{\times}
\newcommand{\cartesianProduct}[2]{\ensuremath{#1 \cartesianC #2}}

\newcommand{\adjunt}[1]{\ensuremath{ #1^{\dagger} }}

\newcommand*{\Tr}[1]{\ensuremath{  \mathsf{Tr}[ #1 ]} }

\makeatletter

\newcommand*{\OPTMath}{\ensuremath{\Theta}}

\newcommand{\rbra}[1]{({#1}\vert}
\newcommand{\rket}[1]{\vert{#1})}

\newcommand{\bra}[1]{\left\<{#1}\right\vert}
\newcommand{\ket}[1]{\left\vert{#1}\right\>}

\newcommand{\ketbra}[2]{\ket{#1}\bra{#2}}

\newcommand{\uniDetEff}{e}

\newcommand{\system}[1]{\ensuremath{\mathrm{#1}}}

\newcommand{\trivialSystem}{\ensuremath{\system{I}}}
\newcommand{\s}[1]{\ustick{\scriptstyle{\system{#1}}}}%{\ustick{\scriptstyle{\system{#1}}}}

\newcommand*{\Sys}[1]{\ensuremath{\mathsf{Sys\left(\mathrm{#1}\right)}}}
\newcommand*{\@sysDimensionD}{\mathsf{D}}
\newcommand*{\sysDimension}[1]{\ensuremath{\@sysDimensionD_{\system{#1}}}}

\makeatletter
\newcommand*\bigcdot{\mathpalette\bigcdot@{.5}}
\newcommand*\bigcdot@[2]{\mathbin{\vcenter{\hbox{\scalebox{#2}{$\m@th#1\bullet$}}}}}
\makeatother

\newcommand*{\outcomeSpace}[1]{\ensuremath{\mathsf{#1}}}
\newcommand*{\outcome}[1]{\ensuremath{#1}}
\newcommand*{\outcomeDouble}[2]{\left(\outcome{#1}, \outcome{#2} \right)}
\newcommand*{\outcomeIncluded}[2]{\ensuremath{ \outcome{#1} \in \outcomeSpace{#2} }}
\newcommand*{\outcomeSpaceDouble}[2]{\ensuremath{\cartesianProduct{\outcomeSpace{#1}}{\outcomeSpace{#2}}}}
\newcommand*{\outcomeIncludedDouble}[4]{\ensuremath{ \outcomeDouble{#1}{#2} \in \outcomeSpaceDouble{#3}{#4} }}
\newcommand*{\outcomeSpaceConditioned}[2]{\ensuremath{\outcomeSpace{#1}^{\left( \outcome{#2} \right)}}}
\newcommand*{\outcomeIncludedConditioned}[3]{\ensuremath{\outcome{#1} \in \outcomeSpaceConditioned{#2}{#3}}}

\newcommand{\testNoDown}[1]{\ensuremath{\mathsf{#1}}}

\newcommand{\test}[2]{\ensuremath{\testNoDown{#1}_{\outcomeSpace{#2}}}}
\newcommand{\conditionedTest}[3]{\ensuremath{\testNoDown{#1}^{\left(\outcome{#3}\right)}_{\outcomeSpace{#2}}}}

\newcommand{\preparationTestNoDown}[1]{{\ensuremath{#1}}}

\newcommand{\eventNoDown}[1]{\ensuremath{\mathscr{#1}}}
\newcommand{\event}[2]{\ensuremath{\eventNoDown{#1}_{\outcome{#2}}}}
\newcommand{\eventCG}[2]{\ensuremath{\eventNoDown{#1}_{\outcomeSpace{#2}}}}
\newcommand{\preparationEventNoDown}[1]{\ensuremath{#1}}

\newcommand{\observationEventNoDown}[1]{\ensuremath{\text{#1}}}
\newcommand{\observationEvent}[2]{\ensuremath{\observationEventNoDown{#1}_{\outcome{#2}}}}
\newcommand{\observationUniqueDeterministic}{\observationEventNoDown{\uniDetEff}}

\newcommand{\probabilityEventNoDown}[1]{\ensuremath{ {#1} }}
\newcommand{\probabilityP}{\probabilityEventNoDown{p}}

\newcommand{\conditionedEvent}[3]{\ensuremath{\eventNoDown{#1}^{\left(\outcome{#3}\right)}_{\outcome{#2}}}}
\newcommand{\conditionedEventTest}[4]{\{\eventNoDown{#1}^{\left(\outcome{#4}\right)}_{\outcome{#2}}\}_{\outcomeIncluded{#2}{#3}}}
\newcommand{\preparationEventTest}[3]{\ensuremath{\{\preparationEventNoDown{#1}_{\outcome{#2}} \}_{\outcomeIncluded{#2}{#3}}}}
\newcommand{\observationEventTest}[3]{\ensuremath{\{\observationEventNoDown{#1}_{\outcome{#2}} \}_{\outcomeIncluded{#2}{#3}}}}
\newcommand*{\eventNoDownWithIdentityParentesys}[2]{\ensuremath{ \left( \eventNoDownWithIdentityParentesys{#1}{#2} \right) }}

\newcommand{\nullTransformationSymbol}{\varepsilon}
\newcommand{\nullTransformation}[2]{\nullTransformationSymbol_{\system{#1} \!\to\! \system{#2}}}

\newcommand{\seqC}{\circ}
\newcommand{\sequentialComp}[2]{{#1} \seqC {#2}}
\newcommand{\sequentialEventTest}[6]{\{ \sequentialComp{\mathcal{#4}_{\outcome{#5}}}{\mathcal{#1}_{\outcome{#2}}} \}_{\outcomeIncludedDouble{#2}{#5}{#3}{#6}}}

\newcommand*{\setNotEnsemble}[1]{ \ensuremath{\left\llbracket #1 \right\rrbracket} }
\newcommand{\eventTestUnformatted}[3]{\ensuremath{\left\{ #1  \right\} _{\outcomeIncluded{#2}{#3}}}}
\newcommand{\eventTest}[3]{\ensuremath{ \eventTestUnformatted{ \eventNoDown{#1}_{\outcome{#2}} }{#2}{#3} }}
\newcommand{\paralC}{ \boxtimes }
\newcommand{\parallelComp}[2]{{#1} \paralC {#2}}
\newcommand{\parallelEventTest}[6]{\{ \parallelComp{\mathcal{#1}_{\outcome{#2}}}{\mathcal{#4}_{\outcome{#5}}} \}_{\outcomeIncludedDouble{#2}{#5}{#3}{#6}}}

\newcommand{\St}[1]{\ensuremath{\mathsf{St\!\left(\system{#1}\right)}}}	
\newcommand{\StN}[1]{\ensuremath{\mathsf{St_{1}\!\left(\system{#1}\right)}}}				%Deterministic
\newcommand{\verSt}[1]{\ensuremath{\Upsilon\!\left(#1\right)}}

\newcommand{\Eff}[1]{\ensuremath{\mathsf{Eff\!\left(\system{#1}\right)}}}					    %Effects
\newcommand{\EffN}[1]{\ensuremath{\mathsf{Eff_{1}\!\left(\system{#1}\right)}}}				%Deterministic

\newcommand{\Transf}[2]{\ensuremath{\mathsf{Transf\!\left(\system{#1}\!\to\!\system{#2}\right)}}}

\newcommand{\Instr}[2]{\ensuremath{\mathsf{Instr\!\left(\system{#1}\!\to\!\system{#2}\right)}}}

\newcommand{\InstrA}[1]{\ensuremath{\mathsf{Instr\left(\system{#1}\right)}}}

\newcommand{\InstrOPT}[1]{\ensuremath{\mathsf{Instr\left({#1}\right)}}}

\newcommand{\Prep}[1]{\ensuremath{\mathsf{Prep\!\left(\system{#1}\right)}}}
\newcommand{\Obs}[1]{\ensuremath{\mathsf{Obs\!\left(\system{#1}\right)}}}

\makeatother

\newcommand{\doNotExcludeOp}{\ensuremath{\rightarrow}}
\newcommand{\doNotExclude}[2]{\ensuremath{#1 \doNotExcludeOp #2}}

\newcommand*{\entropy}[2][]{\ensuremath{H_{#1}\left(  #2 \right)}}

\newcommand{\secRef}[1]{\hyperref[#1]{section~\ref*{#1}}}
\newcommand{\aref}[1]{\hyperref[#1]{Appendix~\ref*{#1}}}

\usepackage{acronims}

\begin{document}
	
	\title{Complementarity in quantum and general information theories}
	
	\author{Davide Rolino}
	\email{davide.rolino01@universitadipavia.it}
	\affiliation{yiyaniQ, Toronto, ON, Canada}
	\affiliation{Department of Physics \& Astronomy, University of Waterloo, Waterloo, Ontario, Canada N2L 3G1}
	\affiliation{Perimeter Institute for Theoretical Physics, 31 Caroline Street North, Waterloo, Ontario, Canada N2L 2Y5}

	\author{Paolo Perinotti}
	\email{paolo.perinotti@unipv.it}
	\affiliation{Dipartimento di Fisica, Università degli Studi di Pavia, Via Bassi 6, 27100 Pavia, Italia}
	\affiliation{INFN Gruppo IV, Sezione di Pavia, Via Bassi 6, 27100 Pavia, Italia}
	
	\author{Alessandro Tosini}
	\email{alessandro.tosini@unipv.it}
	\affiliation{Dipartimento di Fisica, Università degli Studi di Pavia, Via Bassi 6, 27100 Pavia, Italia}
	\affiliation{INFN Gruppo IV, Sezione di Pavia, Via Bassi 6, 27100 Pavia, Italia}
	
	\begin{abstract}
		We propose an operational definition of complementarity, pinning down the concept originally introduced by Bohr. Two properties of a system are considered complementary if the definiteness of their values conflicts. We show that, within quantum theory, this notion is equivalent to the incompatibility of operations, that is, their inability to be performed simultaneously. We further show that for a generic information theory, the instruments associated with complementary properties are incompatible.
	\end{abstract}
	
	\maketitle
	
	\section{Introduction}
	In the early days of quantum mechanics, some counterintuitive features of the theory
	were not yet thoroughly formalised, and physicists struggled to grasp them in the  
	classical language. A paradigmatic example is that of Bohr's concept of \emph{complementarity}~\cite{bohrQuantumPostulateRecent1928,bohrCanQuantumMechanicalDescription1935,pauliGeneralPrinciplesQuantum1980,Saunders:2005aa}. A prominent manifestation of complementarity is the so-called wave-particle duality, a simplified phrasing for the observation that neither the corpuscular nor the wave-like behaviour of a quantum particle fully captures its physics. In the modern language of quantum field theory, a thorough mathematical formalism allows one to encompass both of these complementary aspects of microscopic systems; however, the interpretative challenges are only shifted to a different level. 
	
	Following Bohr's original idea, an extensive line of research has developed around wave-particle 
	duality~\cite{englertFringeVisibilityWhichWay1996,scullyQuantumOpticalTests1991,jaegerTwoInterferometricComplementarities1995,durrOriginQuantummechanicalComplementarity1998,bimonteInterferometricDualityMultibeam2003,martinez-linaresQualityWhichwayDetector2004,schillingPhasedependentWhichwayInformation2012,dezelaCommentFringeVisibility2013,prabhutejQuantumWhichwayInformation2014,krausComplementaryObservablesUncertainty1987,buschOperationalQuantumPhysics1995,buschComplementarityQuantumObservables1995,buschHeisenbergsUncertaintyPrinciple2007,petzComplementarityQuantumSystems2007,plotnitskyWhatComplementarityNiels2014,sahaOperationalFoundationsComplementarity2020,hsiehQuantumComplementarityNovel2023,serinoComplementaritybasedComplementarity2024}. Traditionally, complementarity is a notion associated with \emph{physical quantities} that can be measured on a quantum system~\cite{krausComplementaryObservablesUncertainty1987,buschOperationalQuantumPhysics1995,buschComplementarityQuantumObservables1995,buschHeisenbergsUncertaintyPrinciple2007,petzComplementarityQuantumSystems2007,plotnitskyWhatComplementarityNiels2014,sahaOperationalFoundationsComplementarity2020,hsiehQuantumComplementarityNovel2023,serinoComplementaritybasedComplementarity2024}. However, in the context of quantum information, the physical interpretation of measurement outcomes plays a secondary role, and a notion such as complementarity must be defined in a more abstract manner. We address this question by defining complementarity in terms of the impossibility of simultaneously assigning truth values to statements concerning the outcomes of particular measurements, which we deem \emph{properties} of a physical system.
	
	Since a property is some predicate of a system that has to exhibit some degree of objectivity, or better, intersubjectivity, in its very definition we will require not only that its values might have truth values (corresponding to probability 0 or 1), as is usual in the literature, but also that a repeated measurement would not provide contradictory assertions. Actually, the former feature can be derived as a consequence of the latter. In this way, the dynamical aspects become important in the distinction between any measurement and the measurement of a property. Once we have a satisfactory definition of property at hand, we introduce the notion of \textdef{complementary} properties. This notion captures the idea that there is a conflict between the definiteness of values of two properties: If one has definite values, those of the other need to be undefined.
	
	In Section~\ref{s:quantum} we focus on quantum theory and show that complementarity is equivalent to the \emph{incompatibility} of the instruments associated with the measurement of such properties, namely the impossibility of performing them simultaneously~\cite {buschComparingDegreesIncompatibility2013,heinosaariSimultaneousMeasurementTwo2016,heinosaariInvitationQuantumIncompatibility2016,darianoIncompatibilityObservablesChannels2022,buscemiUnifyingDifferentNotions2023}. In Section~\ref{s:general} we provide a definition of complementarity that holds for arbitrary information theories (even non-causal) and later 
	demonstrate that for any causal theory complementary measurements are incompatible processes~\cite{hardyDisentanglingNonlocalityTeleportation1999,barnumCloningBroadcastingGeneric2006,barrettInformationProcessingGeneralized2007,spekkensEvidenceEpistemicView2007,chiribellaProbabilisticTheoriesPurification2010,janottaGeneralizedProbabilisticTheories2013,darianoQuantumTheoryFirst2016,plavalaGeneralProbabilisticTheories2021}. The other implication, namely that incompatible operations lead to complementary properties, strongly relies on the mathematical characterization of quantum properties. In Section~\ref{sec:conclusion} we summarise and discuss prospects for future work.
	
	\section{Quantum complementarity}\label{s:quantum}
	Let $\system A$ be a quantum system with Hilbert space $\hilbert{} $, and let $\test{T}{X}=\left\{ \event{T}{x} \right\}_{\outcomeIncluded{x}{X}} $ be a \emph{quantum instrument} for $\system A$, that is a collection of \emph{quantum operations} $\event{T}{x}$, i.e.~completely positive and trace-nonincreasing maps, such that their sum is a \emph{channel}, i.e.~a completely positive and trace-preserving map.
	
	We will also represent instruments and quantum operations graphically via wired boxes
	\begin{align*}
		\myQcircuit{
			&\s{A}\qw&\gate{\left\{ \event{T}{x} \right\}_{\outcomeIncluded{x}{X}}}&\s{B}\qw&\qw&
		}, \quad \myQcircuit{
			&\s{A}\qw&\gate{\event{T}{x}}&\s{B}\qw&\qw&
		},
	\end{align*}
	where input and output wires depict the input and output systems.
	
	In particular, we will consider instruments with single Kraus quantum operations, i.e.
	\begin{align}\label{eq:atinst}
		\event{T}{x}\left( \preparationTestNoDown{\rho} \right)= T_{\outcome{x}} \preparationTestNoDown{\rho} \adjunt{T_{\outcome{x}}},\quad\forall \outcomeIncluded{x}{X}.
	\end{align}
	We recall that single Kraus quantum operations $\eventNoDown{T}\left( \cdot \right) = T \cdot \adjunt{T}$ belong to the extremal rays of the cone of completely positive maps. Therefore, they cannot be decomposed as a coarse-graining (sum) of other operations (see Section~\ref{s:opts} for a formal definition of coarse-graining of a test). The only allowed decomposition is of the form: $\eventNoDown{T} = \sum_{\outcomeIncluded{x}{X}} \probabilityP_{\outcome{x}} \eventNoDown{T}$, where $\left\{ \probabilityP_{\outcome{x}} \right\}_{\outcomeIncluded{x}{X}}$ is a probability distribution. Following the literature on operational theories, we refer to these transformations as \emph{atomic operations}, and an \emph{atomic instrument} as a collection of atomic operations.
	
	We now define an instrument \emph{repeatable} if applying it multiple times to the same system is equivalent to applying it once. A Stern-Gerlach apparatus is a typical example: Repeated applications with the same magnetic-field orientation will provide the same outcome. Formally (see also Definition~\ref{def:opt:instr:repeatable}), this requirement is expressed as
	\begin{align}\label{eq:repeat}
		\event{T}{x} \event{T}{x'} = \kronekerDelta{x}{x'} \event{T}{x} \quad \forall \outcome{x}, \outcome{x'} \in \outcomeSpace{X}.
	\end{align}
	
	We now define \emph{elementary properties} as repeatable atomic quantum instruments.
	The repeatability requirement translates the idea that a measurement of the property leaves no detectable trace in the system and not only that the probability of reading the same outcome in two subsequent measurements is one. This is what we expect from the notion of a property, which should be to the largest possible extent attributed to the system rather than to the observation process.
	Clearly, due to the condition~\eqref{eq:repeat}, each event admits a state $\preparationEventNoDown{\rho}^{\left(\outcome{x}\right)}$ that occurs with probability one: $ \Tr{\event{T}{x} \left(\preparationEventNoDown{\rho}^{\left(\outcome{x}\right)}\right)} = 1$, which can be obtained as $\preparationEventNoDown{\rho}^{\left(\outcome{x}\right)}=\event{T}{x} \left(\preparationEventNoDown{\rho}\right)/\Tr{\event{T}{x} \left(\preparationEventNoDown{\rho}\right)}$ for any state $\rho$.
	We refer to such a state as a \emph{verifier} for the event $\event{T}{x}$, and more generally, for the property $\test{T}{X}$. We denote the set of verifiers of the operation $\event{T}{x}$ by $\verSt{\event{T}{x}}$ and that of the property $\test{T}{X}$ by $\verSt{\test{T}{X}}$. We can also introduce the notion of a \textdef{property}, which differs from that of elementary property in that it can be obtained as a coarse graining of an elementary property, where one or more outcomes can be obtained by merging two or more outcomes of the elementary property, e.g.~$\event{T}{x}=\eventNoDown{T}'_{\outcome{y}_{1}} + \eventNoDown{T}'_{\outcome{y}_{2}}$. We notice that an elementary quantum property has two main features: It can be verified under suitable circumstances (existence of verifier states), and it is objective (repeatable). 
	
	In Section~\ref{s:opts} we define elementary properties for a general information theory. While for causal theories, as the quantum one, any repeatable test allows for a verifier, in non-causal theories this is not the case. For this reason verifiability will be assumed in Definition~\ref{def:elementary-property} as an elementary property of a general theory. In addition to being always verifiable, instruments that correspond to quantum elementary properties can be easily characterized in terms of commuting projective measurements, as expected.
	
	\begin{lemma}[Quantum elementary property]\label{lem:quantum-elementary-properties}
		An instrument $\test{T}{X}=\left\{ \event{T}{x} \right\}_{\outcomeIncluded{x}{X}} $ corresponding to an elementary property has the form
		\begin{align}\label{eq:projectors}
			\event{T}{x}(\preparationEventNoDown{\rho})=\Pi_x\preparationEventNoDown{\rho}\Pi_x,\quad\Pi_x\Pi_{x'}=\delta_{xx'}\Pi_x,\quad\sum_{\outcomeIncluded{x}{X}}\Pi_x=I.
		\end{align}
		where $\Pi_{\outcome{x}}$ is a projector.
	\end{lemma}
	\begin{proof}
		Consider a repeatable atomic instrument as in Eq.~\eqref{eq:atinst}. From the repeatability condition~\eqref{eq:repeat}, one has $\eventNoDown{T}_{\outcome{x}}^2=\eventNoDown{T}_{\outcome{x}}$; therefore, it follows that 
		\begin{align*}
			T_{\outcome{x}}^2 = e^{i \phi} T_{\outcome{x}}\:\Rightarrow \:T_{\outcome{x}} = e^{i \phi} \Pi_{\outcome{x}},
		\end{align*} 
		where $\Pi_{\outcome{x}}$ is the projection on the support of $T_x$, as can be shown
		applying the pseudoinverse of $T_{\outcome{x}}$ to both sides of the first equation. 
		The global phase factor plays no role, since Kraus operators appear only in expressions of the form $T_{\outcome{x}} \preparationTestNoDown{\rho} \adjunt{T_{\outcome{x}}}$; hence we ignore it. Therefore, for an atomic repeatable instrument we have
		\begin{align*}
			\event{T}{x}(\rho)=\Pi_x\rho\Pi_x.
		\end{align*}
		Moreover, given two such operations in the same instrument, by the condition~\eqref{eq:repeat} we have ${\Pi}_{x}{\Pi}_{x'}=0$, from which the thesis follows.
	\end{proof}
	
	From the above characterization of quantum elementary properties~\eqref{eq:projectors}, which generally does not hold for a generic information theory, an additional property of quantum verifiers follows: All quantum verifiers are fixed points of the property.
	\begin{lemma}[Quantum verifiers]\label{lem:quantum-verifiers}
		The set of verifier states $\preparationEventNoDown{\rho}^{(\outcome{x})} \in \verSt{\event{T}{x}}$ of any quantum elementary property $\test{T}{X} \equiv \eventTest{T}{x}{X}$ is such that $\event{T}{x} ({\preparationEventNoDown{\rho}^{\left( \outcome{x} \right)}}) = {\preparationEventNoDown{\rho}^{\left( \outcome{x} \right)}}$. 
	\end{lemma}
	\begin{proof}
		This follows joining the verifier condition $\Tr{\event{T}{x} (\preparationEventNoDown{\rho}^{\left(\outcome{x}\right)})} = 1$ and the fact that $\event{T}{x}$ is projection, whose support then has to include the state $\rho^{\left( \outcome{x} \right)}$.
	\end{proof}
	
	Notice that the converse is always true: Fixed points of the property are verifiers.
	
	The above result, holding in quantum theory, grants a property of great importance that we call \emph{strong verifiability}. The definition of strong verifiability in general information theories is given in Remark~\ref{rmk:opt:state:verifier:strong}. Strong verifiability plays a key role in the main result of this work.
	
	Having characterised the elementary properties of quantum theory, we are now in position to define two elementary properties as complementary if there exists a verifier state for one that is not a verifier state for the other (see also Definition~\ref{def:complementarity}). If we prepare such a verifier state, right after the measurement of the corresponding property we have a situation where the verified property is well defined, while the other is undefined. This cannot arise classically: The pure states of a classical system are identified with the verifiers of the unique elementary property (point in the phase space), which is thus well defined. 
	
	\begin{remark} In the definition of complementarity we consider elementary properties. This follows from an operational \emph{desideratum}: Complementarity must be unrelated to the noise introduced in the measurement by ignoring features through coarse-graining events. Consider, for example, the  two non-elementary properties for a qutrit
		\begin{align*}
			&  \left\{ \Pi_0\cdot \Pi_0+\Pi_1\cdot \Pi_1,\Pi_2\cdot \Pi_2\right\},\\
			&  \left\{ \Pi_+\cdot \Pi_++\Pi_-\cdot \Pi_-,\Pi_2\cdot \Pi_2\right\},
		\end{align*}
		where $\Pi_i=\ketbra{i}{i}$, and $\{\ket{i}\}_ {i=0,1,2}$ denotes the canonical basis while $\ket{\pm}=(\ket{0}\pm\ket{1})/\sqrt{2}$. It is immediate to see that these two properties have all verifier states in common. We could define such a pair of tests as non-complementary. However, this is not completely satisfactory: Complementarity does not manifest only because information from a more refined measurement is ignored. Non-complementarity of this kind is not of interest, given that it arises only due to coarse graining.
	\end{remark}
	
	Besides the on-off notion of complementarity given above, one can further refine it into more quantitative ones. To this end,
	in Definition~\ref{def:opt:compl:degree} we will introduce three decreasingly strong variants of complementarity for an arbitrary theory of information. Lemma~\ref{lem:causal_opt:compl:degree} specifies these degrees for the particular case of causal theories. Since quantum theory belongs to this class, the same degrees of complementarity hold for it as well.
	
	The notion of \textdef{incompatibility} of measurements and channels has been largely investigated in the literature~\cite{buschComparingDegreesIncompatibility2013,heinosaariSimultaneousMeasurementTwo2016,heinosaariInvitationQuantumIncompatibility2016} and recently adapted to instruments in Refs.~\cite{darianoIncompatibilityObservablesChannels2022,buscemiUnifyingDifferentNotions2023}. In the case of instruments, we are interested in not only determining whether two operations can be performed simultaneously, but if implementing one prevents in some way a subsequent implementation of the other. Indeed, one has to recover not only the classical outcomes of both measurements but also the corresponding reduced quantum states through generalised post-processing. We term this notion of incompatibility \textdef{strong incompatibility} or, equivalently, \textdef{weak compatibility}. The notion of weak compatibility is based on the property of \emph{exclusion}. Formally (see also Section~\ref{s:compatibility}), we say that an operation $\test{T}{X} \equiv \eventTest{T}{x}{X}$ does not exclude the operation $\test{G}{Y} \equiv \eventTest{G}{y}{Y}$ if it is possible to implement the first operation in such a way that it is possible to post-process the outcome so that comprehensively the second operation is implemented. In formulas this is expressed by the fact that the equalities \begin{align}
		\label{eqt:doesNotExclude-1}
		&\myQcircuit{
			&\s{A}\qw&\gate{\event{T}{x}}&\s {B}\qw&\qw&
		} = \sum_{\outcomeIncludedConditioned{z}{S}{x}}
		\myQcircuit{
			&\s{A}\qw&\multigate{1}{\event{C}{z}}&\qw&\s {B}\qw&\qw&
			\\
			&\pureghost{}&\pureghost{\event{C}{z}}&\s{E}\qw&\measureD{\observationUniqueDeterministic}
		} \; \forall \outcome{x},\\\label{eqt:doesNotExclude-2}
		&\myQcircuit{
			&\s{A}\qw&\gate{\event{G}{y}}&\s{C}\qw&\qw&
		} = \sum_{\outcomeIncluded{z}{Z}}
		\myQcircuit{
			&\s{A}\qw&\multigate{1}{\event{C}{z}}&\s {B}\qw&\multigate{1}{\conditionedEvent{P}{y}{z}}&\s{C}\qw&\qw&
			\\
			&\pureghost{}&\pureghost{\event{C}{z}}&\s{E}\qw&\ghost{\conditionedEvent{P}{y}{z}}&\pureghost{}&
		} \; \forall \outcome{y}
	\end{align}
	are satisfied for some instrument $\eventTest{C}{z}{Z} \in \Instr{A}{BE}$ and a post-processing, i.e.~a family of instruments $\{\conditionedEventTest{P}{y}{Y}{z}\}_{\outcomeIncluded{z}{Z}} \subset \Instr{BE}{C}$. Then two operations are weakly compatible if this relation is satisfied also exchanging the roles of $\eventTest{T}{x}{X}$ and $\eventTest{G}{y}{Y}$, possibly for a different instrument and post-processing.
	In the above realization scheme an arbitrary ancillary system $\system{E}$ is allowed, which is later discarded via the partial trace (the symbol $e$ at output wire $\system{E}$). 
	
	To get a more intuitive idea of the notion of incompatibility, one can think about how the process of measuring the position of a particle changes its momentum, for example, as proposed in Heisenberg's \emph{gedankenexperiment}~\cite{heisenbergUeberAnschaulichenInhalt1927}. Measuring the position alters the momentum of the particle, making it impossible to perform a measurement that simulates a momentum measurement on the particle \emph{before the position measurement}. We highlight that in the case of quantum theory all quantum channels are compatible, i.e.~application of one channel is not incompatible with a subsequent post-processing that simulates another channel, exploiting the fact that all channels have a unitary dilation. Incompatibility, on the other hand, is manifest in quantum instruments~\cite{darianoIncompatibilityObservablesChannels2022}.
	
	We now state our main result: In quantum theory complementarity is equivalent to incompatibility, or analogously non-complementarity is equivalent to compatibility
	
	\begin{theorem}In quantum theory non-complementarity is equivalent to compatibility: Two elementary properties
		$\test{T}{X} \equiv \eventTest{T}{x}{X}$ and $\test{G}{Y} \equiv \eventTest{G}{y}{Y}$ are compatible if and only if they have the same set of verifier states, namely $\verSt{\test{G}{Y}} = \verSt{\test{T}{X}}$.
	\end{theorem}
	
	\begin{proof}
		We start by looking at the implication: Compatibility implies non-complementarity, namely if $\test{T}{X} \equiv \eventTest{T}{x}{X}$ and $\test{G}{Y} \equiv \eventTest{G}{y}{Y}$ are compatible, then they have the same set of verifier states. In particular, the result follows from the fact that if $\test{T}{X}$ does not exclude $\test{G}{Y}$, then $\verSt{\test{G}{Y}} \subseteq \verSt{\test{T}{X}}$ and from the opposite inclusion obtained reversing the roles of tests. The proof of this fact, sketched here, is provided for general information theories in Theorem~\ref{thm:opt:verSt:primaryInclusionProperty}. Proving the relationship is sufficient to show that if $\preparationEventNoDown{\rho}$ is a verifier state for $\event{G}{y}$, then $\Tr{\event{G}{y}\preparationEventNoDown{\rho}} \leq \Tr{\event{T}{x}\preparationEventNoDown{\rho}}$. Since the quantum operations of elementary properties are atomic, the compatibility condition~\eqref{eqt:doesNotExclude-2} gives 
		\begin{align*}
			\myQcircuit{
				&\s{A}\qw&\multigate{1}{\event{C}{z}}&\qw&\s {A}\qw&\qw&
				\\
				&\pureghost{}&\pureghost{\event{C}{z}}&\s{E}\qw&\measureD{\observationUniqueDeterministic}
			}
			\propto 
			\myQcircuit{
				&\s{A}\qw&\gate{\event{T}{x}}&\s {A}\qw&\qw&
			} 
			\quad \forall \outcomeIncludedConditioned{z}{S}{x}.
		\end{align*}
		From this equation and the fact that the instrument $\test{T}{X}$ is repeatable, it can be shown that for any $\outcomeIncluded{y}{Y}$ there exists a unique $\outcomeIncluded{x}{X}$ such that 
		\begin{align*}
			\event{G}{y} \equiv \event{D}{y,x} = \sum_{\outcomeIncludedConditioned{z}{Z}{x}}
			\myQcircuit{
				&\s{A}\qw&\multigate{1}{\event{C}{z}}&\s {A}\qw&\multigate{1}{\conditionedEvent{P}{y}{z}}&\s{A}\qw&\qw&
				\\
				&\pureghost{}&\pureghost{\event{C}{z}}&\s{E}\qw&\ghost{\conditionedEvent{P}{y}{z}}&\pureghost{}&
			}.
		\end{align*}
		One can then notice that $ \sum_{\outcomeIncludedConditioned{z}{Z}{x}} \Tr{ \conditionedEvent{P}{y}{z} \event{C}{z} \preparationEventNoDown{\rho}} \leq \sum_{\outcomeIncludedConditioned{z}{Z}{x}} \Tr{ \event{C}{z} \preparationEventNoDown{\rho}} = \Tr{\event{T}{x} \preparationEventNoDown{\rho}}$, where first bound  follows from the fact that $\conditionedEvent{P}{y}{z}$ are generally not deterministic. We have therefore recovered the desired inequality, showing that if two properties are weakly compatible, they must also be non-complementary. 
		
		We now prove the converse statement, namely that in quantum theory non-complementarity implies compatibility. The proof relies on the special form of quantum elementary properties proved in Lemma~\ref{lem:quantum-elementary-properties}, and on the consequent property satisfied by their verifier sets, given in Lemma~\ref{lem:quantum-verifiers}. Taking this for granted, we consider two non-complementary properties $\test{G}{Y} \equiv \eventTest{G}{y}{Y}$ and $\test{T}{X}\equiv \eventTest{T}{x}{X}$ and show that they must  correspond to the same operation, up to a permutation of their elements. As a consequence, they will trivially be compatible. The two instruments share the same verifier set by hypothesis: $\verSt{\test{T}{X}} = \verSt{\test{G}{Y}}$. Furthermore, for any $\outcome{x} \in X$, there must exist $\outcome{y} \in Y$ such that $\verSt{\event{T}{x}} = \verSt{\event{G}{y}}$; otherwise the equality of the global verifier would be violated. Suppose now that two states in $\verSt{\event{T}{x}}$ belong to $\verSt{\event{G}{y}}$ and $\verSt{\event{G}{y'}}$, respectively, with $\outcome{y} \neq \outcome{y'}$. Then their convex combination would belong to $\verSt{\test{T}{X}}$, but not to $\verSt{\test{G}{Y}}$, which contradicts the hypothesis. By the same argument  the cardinalities of $\outcomeSpace{X}$ and $\outcomeSpace{Y}$ must be the same. Exploiting this and the fact that the Kraus operators are orthogonal projections decomposing identity, we conclude that the supports of such projections must be the same for the two properties, their supports being those of joint verifier states. Since in quantum theory an orthogonal projector is uniquely determined by its support, we obtain 
		$\Pi^{\testNoDown{P}}_{\outcome{x}} = \Pi^{\testNoDown{Q}}_{\varphi(\outcome{x})}$
		for a bijection $\varphi:\outcomeSpace{X}\to\outcomeSpace{Y}$. 
		Hence the two instruments coincide.
	\end{proof}
	
	The equivalence established in the above theorem is a genuinely quantum result, enabled by structural features of quantum theory that do not hold in arbitrary operational probabilistic theories (OPTs). The direction compatibility implies non-complementarity follows from Theorem~\ref{thm:opt:verSt:primaryInclusionProperty} below, which holds for any causal OPT. The converse direction, non-complementarity implies compatibility, instead rests on three properties that are specific to quantum theory: (i) Elementary properties are projective measurements (Lemma~\ref{lem:quantum-elementary-properties}), a characterization that follows from the interplay of atomicity and repeatability within the Hilbert space formalism; (ii) quantum verifier states are strong verifiers, i.e.~fixed points of the corresponding elementary property (Lemma~\ref{lem:quantum-verifiers}), as a direct consequence of the projective structure; and (iii) an orthogonal projector is uniquely determined by its support, so two elementary properties sharing the same verifier set must coincide. None of these features are guaranteed in a generic causal OPT, where elementary properties need not admit a projective representation and verifiers need not be fixed points of the instrument. As a result, in the next section  we show that for generic causal OPTs only one direction of the equivalence survives: Complementarity implies incompatibility (Corollary~\ref{cor:complementarity-imnplies-incompatibility}).
	
	\section{Complementarity in Operational Probabilistic Theories}\label{s:general}
	
	\subsection{Operational Probabilistic Theories}\label{s:opts}
	
	Here we aim at giving a brief introduction to the framework of operational probabilistic theories~\cite{chiribellaProbabilisticTheoriesPurification2010,chiribellaInformationalDerivationQuantum2011,darianoQuantumTheoryFirst2016,chiribellaQuantumPrinciples2016,darianoClassicalityLocalDiscriminability2020,perinottiCellularAutomataOperational2020} in order to prove that the result that complementarity implies incompatibility can be generalised also to generic theories of information processing.
	
	The purpose of the framework is to be able to describe and study all theories that have a compositional structure analogous to that of quantum (and classical) theory. Indeed, the way physical processes can be composed sequentially or in parallel is the same for all these theories. One of the main objectives of the framework is to classify these theories with respect to the possibility or impossibility of performing specific information processing tasks, e.g., the possibility of making simultaneous statements about properties of physical systems and studying logical dependencies among such features.
	
	The construction of an \ac{OPT} is based on three fundamental concepts: systems, instruments and transformations, and probabilities.
	
	\textdef{Systems}, or better \emph{system types}, symbolically denoted by uppercase Latin letters \system{A}, \system{B}, $\ldots$, are an abstract notion that is used to keep track of causal connections between  physical events. Examples of a system type include a mode of the electromagnetic field, a mode of the Dirac field, and, in non-relativistic quantum mechanics, the electron spin and a system of two electrons.  
	\textdef{Instruments}, $\test{T}{X} \equiv \eventTest{T}{x}{X} \in \Instr{A}{B}$, instead, model the experiments or processes that can occur in systems. Every instrument has an input system \system{A} and an output system \system{B}, while 
	$\outcomeSpace{X}$ is the \textdef{outcome space} associated with the experiment, which is the set of all possible classical outcomes that can be registered as a consequence of running the experiment. Each outcome corresponds to the occurrence of a particular physical event, mathematically modelled through the notion of a \textdef{transformation} $\event{T}{x} \in \Transf{A}{B}$. In a classical setting, an instrument could, for example, model the roll of a die, with every transformation being associated with the die ending up on a particular face, the corresponding number representing the outcome of the experiment.
	
	As previously mentioned, the characterising aspect of \acp{OPT} is their compositional structure, which aims to reproduce that of quantum theory. In particular, we require the collection of instruments, and consequently transformations, to be closed under sequential and parallel composition. Formally, this means that for every pair of instruments $\eventTest{T}{x}{X} \in \Instr{A}{B}$ and $\eventTest{G}{y}{Y} \in \Instr{B}{C}$ where the output system of the former coincides with the input system of the latter, there exists an instrument $\sequentialEventTest{T}{x}{X}{G}{y}{Y} \in \Instr{A}{C}$ in the \ac{OPT} modelling experiments subsequently performed on a physical system. Analogously, in the case of parallel composition, for every pair of instruments $\eventTest{T}{x}{X} \in \Instr{A}{B}$ and $\eventTest{G}{y}{Y} \in \Instr{C}{D}$ there exists $\parallelEventTest{T}{x}{X}{G}{y}{Y} \in \Instr{AC}{BD}$ describing the situation in which two independent instruments are independently performed on the corresponding subsystems.
	
	An added value of the framework is that it allows for a graphical representation of complex networks of processes. Instruments and transformations are represented as wired boxes
	\begin{align*}
		\myQcircuit{
			&\s{A}\qw&\gate{\eventTest{T}{x}{X}}&\s{B}\qw&\qw&
		}, \quad \myQcircuit{
			&\s{A}\qw&\gate{\event{T}{x}}&\s{B}\qw&\qw&
		},
	\end{align*}
	where the input-output direction is assumed to be from the left to the right by convention. Sequential and parallel composition can then be represented as
	\begin{align*}
		\myQcircuit{
			&\s{A}\qw&\gate{\eventNoDown{T}}&\s{B}\qw&\gate{\eventNoDown{G}}&\s{C}\qw&\qw&
		} = \quad\!\! \myQcircuit{
			&\s{A}\qw&\gate{\sequentialComp{\eventNoDown{G}}{\eventNoDown{T}}}&\s{C}\qw&\qw&
		}
	\end{align*}
	and
	\begin{align*}
		\myQcircuit{
			&\s{A}\qw&\gate{\eventNoDown{T}}&\s{B}\qw&\qw&
			\\
			&\s{C}\qw&\gate{\eventNoDown{G}}&\s{D}\qw&\qw&
		} = \quad\!\! \myQcircuit{
			&\s{AC}\qw&\gate{\parallelComp{\eventNoDown{T}}{\eventNoDown{G}}}&\s{BD}\qw&\qw&
		},
	\end{align*}
	respectively.
	
	A particular type of system that is assumed to be always present in any \ac{OPT} is the trivial system, denoted \trivialSystem, which is used to describe physical processes in which there are no input and/or output systems. The trivial system acts as the identity for parallel composition of systems, meaning that $\trivialSystem\system{A} = \system{A}\trivialSystem = \system{A}$ for any system in the \ac{OPT}.
	Physical processes where the input system is trivial are preparations, while those where the output system is trivial are measurements, which are processes that just produce a classical output but no output system is used in subsequent instruments. Preparation-instruments will be denoted by Greek letters, $\preparationEventTest{\rho}{x}{X} \in \Prep{A}$, and observation-instruments by lowercase Latin letters, $\observationEventTest{a}{x}{X} \in \Obs{A}$. The transformations composing a preparation-instrument are called \textdef{states} and their set is denoted by $\St{A}$. The transformations composing observation-instruments, instead, are called \textdef{effects} and their set is represented as $\Eff{A}$. Diagrammatically, the trivial system will not be drawn. Preparation- and observation-instruments will then be represented as
	\begin{align*}
		\myQcircuit{
			&\prepareC{\preparationEventTest{\rho}{x}{X}}&\s{A}\qw&\qw&
		}, \quad
		\myQcircuit{
			&\s{A}\qw&\measureD{\observationEventTest{a}{x}{X}}
		}\:,
	\end{align*}
	respectively.
	
	We can now discuss the last element characterising \acp{OPT}, the possibility of assigning probabilities to closed circuits. The compositional structure described so far only allows one to describe any experiment that can be conceived within a given theory, but it does not allow one to make predictions that are comparable to real-world data, something that we deem essential to have a well-defined notion of a physical theory. The way this is obtained in the framework is through the requirement that any instrument with both input and output systems trivial is a probability distribution. As an example, we have the identification
	\begin{align*}
		&		\myQcircuit{			&\prepareC{\preparationEventTest{\rho}{x}{X}}&\s{A}\qw&\gate{\eventTest{G}{y}{Y}}&\s{B}\qw&\measureD{\observationEventTest{a}{z}{Z}}&
		} \\ &\qquad \qquad= \probabilityP\left(  \outcome{x,y,z} \vert \preparationEventTest{\rho}{x}{X}, \eventTest{G}{y}{Y}, \observationEventTest{a}{z}{Z} \right),
	\end{align*}
	i.e.~the circuit on the left has to be interpreted as the probability of obtaining the outcomes \outcome{x,y}, and \outcome{z} given that the experiments $\preparationEventTest{\rho}{x}{X} \in \Prep{A}$, $\eventTest{G}{y}{Y} \in \Transf{A}{B}$, and $\observationEventTest{a}{z}{Z} \in \Obs{A}$ have been performed. 
	
	An important aspect of the framework that will be extensively used in our argument is that it is always possible to coarse grain results, grouping outcomes together.
	
	\begin{definition}[Coarse-graining]\label{def:coarse-graining}
		For every test $\test{T}{X} \equiv \eventTest{T}{x}{X}$ and every disjoint partition $\left\{\outcomeSpace{Z}_{\outcome{y}}\right\}_{\outcomeIncluded{y}{Y}}$ of the outcome space $\outcomeSpace{X}$ there exists the test $\testNoDown{T}^{'}_{\outcomeSpace{Y}}$ representing the same operation, where the outcome $\outcomeIncluded{y}{Y}$ stands for the outcome of the test $\test{T}{X}$ belongs to $\outcomeSpace{Z}_{\outcome{y}}$. The transformation $\eventNoDown{T}'_{\outcome{y}} = \sum_{\outcome{x} \in \outcomeSpace{Z}_{\outcome{y}}} \event{T}{x}$ is called \textdef{coarse-grained transformation}.
	\end{definition}
	
	Obviously, given a test \test{T}{X} the full coarse-graining $\eventCG{T}{X} = \sum_{\outcome{x} \in \outcomeSpace{X}} \event{T}{x}$ is \textdef{deterministic}. We call \emph{deterministic} a transformation associated with a test whose outcome space has just one element. A deterministic transformation does not provide information (the associated test has a unique outcome, occurring with certainty) and can represent, for example, the evolution of an open system. In quantum theory, a deterministic transformation is called a \emph{quantum channel}.
	
	As already mentioned, related to the notion of coarse graining, there exists a particular class of transformations, called \textdef{atomic}, that cannot be written as coarse graining of any other transformation different from themselves.
	\begin{definition}[Atomic transformation]
		\label{def:opt:transf:atomic}
		A transformation $\eventNoDown{T} \in \Transf{A}{B}$ is \textdef{atomic} if, given $\eventNoDown{T}_{1}$, $\eventNoDown{T}_{2} \in \Transf{A}{B}$, one has the implication
		\begin{equation*}
			\eventNoDown{T} = \eventNoDown{T}_{1} + \eventNoDown{T}_{2} \implies \eventNoDown{T}_{1}, \eventNoDown{T}_{2} \propto\eventNoDown{T}.
		\end{equation*}
	\end{definition}
	Consequently, we also define \textdef{atomic instruments}.
	\begin{definition}[Atomic instrument]\label{def:opt:test:atomic}
		An instrument $\test{T}{X} \equiv \eventTest{T}{x}{X}$ is called an \textdef{atomic instrument} if its transformations $\event{T}{x}$ are atomic for all $\outcomeIncluded{x}{X}$.
	\end{definition}
	
	We have covered all the fundamental structure of \acp{OPT} necessary to carry out our argument. However, we remark that the framework is more complex than the one presented here and we refer the interested reader to Refs.~\cite{darianoQuantumTheoryFirst2016,darianoClassicalityLocalDiscriminability2020,perinottiCellularAutomataOperational2020,rolinoMinimalOperationalTheories2025}.
	
	\subsection{Causal theories}
	Now that we have set up the basic structure of the framework of \acp{OPT}, we can discuss a property of \textdef{causality} that can be satisfied by a theory and that is of interest for the present paper. 
	
	\begin{definition}[Causal \acp{OPT}]
		A \textdef{causal \ac{OPT}} \OPTMath{} is a theory where every system $\system{A} \in \Sys{\OPTMath}$ admits a unique deterministic effect~\cite{darianoQuantumTheoryFirst2016}. In symbols, for all $\system{A} \in \Sys{\OPTMath}$ one has $\EffN{A}=\{\observationUniqueDeterministic\}$, where $\EffN{A}$ indicates the set of deterministic effects of system \system{A}.
	\end{definition}
	The reason why this property is referred to as causality is that it can be proven to be equal to the property of \emph{no-signalling from the future}, that is, that the probability distribution of preparation-instruments does not depend on the choice of the observation-instruments at their output~\cite{darianoQuantumTheoryFirst2016}. Formally,
	\begin{align*}
		\sum_{\outcomeIncluded{y}{Y}} \probabilityP\left( \outcome{x,y} \vert \preparationEventTest{\rho}{x}{X}, \observationEventTest{a}{y}{Y} \right) &= \probabilityP\left( \outcome{x} \vert \preparationEventTest{\rho}{x}{X}, \observationEventTest{a}{y}{Y} \right) \\ &= \probabilityP\left( \outcome{x} \vert \preparationEventTest{\rho}{x}{X} \right).
	\end{align*}
	
	A consequence of this is that in causal theories the flow of information is fixed from input to output (diagrammatically, from left to right).
	
	A particular operation that can be encoded in the structure of an \ac{OPT} is the possibility of conditioning which operation to perform based on the result of a previous experiment. 
	
	\begin{definition}[Conditional instruments]
		Let \OPTMath{} be an \ac{OPT}, $\test{T}{X} = \eventTest{T}{x}{X} \in \Instr{A}{B}$ a test of the theory, and $\left\{ \conditionedTest{G}{Y}{x} = \conditionedEventTest{G}{y}{Y}{x} \right\}_{\outcomeIncluded{x}{X}} \subset \Instr{B}{C}$ a labelled collection of instruments. In the case where the collection 
		\begin{equation}
			\label{eqt:OPT:condInstrDef}
			\{\event{T}{(x,y)}\}\coloneqq\left\{ \sequentialComp{\conditionedEvent{G}{y}{x}}{\event{T}{x}} \right\}_{\outcomeIncluded{(x\times y)}{X\times Y}},
		\end{equation}
		is a valid instrument, we call it a \emph{conditional instrument}.
	\end{definition}
	
	This kind of operation is sometimes referred to in the literature as \textdef{classical feedback} or \emph{post-processing}~\cite{leppajarviPostprocessingQuantumInstruments2021}.
	
	\begin{remark}
		\label{remark:condInstr}
		In general, the object in~\eqref{eqt:OPT:condInstrDef} may not belong to $\InstrOPT{\OPTMath}$, namely, it may model an operation that it is not actually implementable in the theory. An example of a theory where not every conditional 
		instrument is an actual instrument of the theory is minimal classical theory proposed in Ref.~\cite{erbaMeasurementIncompatibilityStrictly2024}. A theory where all conditioned instruments are allowed is referred to as \emph{strongly causal}~\cite{darianoQuantumTheoryFirst2016,perinottiCellularAutomataOperational2020,rolinoMinimalOperationalTheories2025}. The notion of causality introduced at the beginning of this section is called \emph{weak causality} in contexts where it must be distinguished from strong causality.
	\end{remark}
	
	\subsection{Compatibility}\label{s:compatibility}
	In conclusion of our discussion on the framework of \acp{OPT}, we will provide the precise definition of the notion of weak-compatibility. As mentioned, the notion of weak-compatibility is based on the notion of exclusion between instruments.
	\begin{definition}[Does not exclude]\label{def:exlusion}
		Let \OPTMath\ be a causal \ac{OPT}. We say that an instrument $\eventTest{T}{x}{X} \in \Instr{A}{B}$ \textdef{does not exclude} another instrument $\eventTest{G}{y}{Y} \in \Instr{A}{C}$ if there exists a test $\eventTest{C}{z}{Z} \in \Instr{A}{BE}$ and a post-processing, i.e.~a family of instruments $\left\{\conditionedEventTest{P}{y}{Y}{z}\right\}_{\outcomeIncluded{z}{Z}} \subset \Instr{BE}{C}$ such that
		\begin{subequations}
			\label{eqt:doesNotExclude}
			\begin{align}
				&\myQcircuit{
					&\s{A}\qw&\gate{\event{A}{x}}&\s {B}\qw&\qw&
				} = \sum_{\outcomeIncludedConditioned{z}{S}{x}}
				\myQcircuit{
					&\s{A}\qw&\multigate{1}{\event{C}{z}}&\qw&\s {B}\qw&\qw&
					\\
					&\pureghost{}&\pureghost{\event{C}{z}}&\s{E}\qw&\measureD{\observationUniqueDeterministic}
				},\\[10pt]
				&\myQcircuit{
					&\s{A}\qw&\gate{\event{B}{y}}&\s{C}\qw&\qw&
				} = \sum_{\outcomeIncluded{z}{Z}}
				\myQcircuit{
					&\s{A}\qw&\multigate{1}{\event{C}{z}}&\s {B}\qw&\multigate{1}{\conditionedEvent{P}{y}{z}}&\s{C}\qw&\qw&
					\\
					&\pureghost{}&\pureghost{\event{C}{z}}&\s{E}\qw&\ghost{\conditionedEvent{P}{y}{z}}&\pureghost{}&
				},
			\end{align}
		\end{subequations}
		where $\left\{\outcomeSpaceConditioned{S}{x}\right\}_{\outcomeIncluded{x}{X}}$ is a suitable partition of \outcomeSpace{X}~\cite{darianoIncompatibilityObservablesChannels2022}. On the other hand, if the above condition fails, we say that the instrument $\eventTest{T}{x}{X} $ \textdef{excludes} $\eventTest{G}{y}{Y} $.  \end{definition}
	
	Consequently, we can define weakly-compatible instruments in the following way.
	\begin{definition}[Weakly-compatible instrument]\label{def:weak-compatibility}
		Let \OPTMath\ be a causal \ac{OPT}. We say that two instruments $\eventTest{T}{x}{X} \in \Instr{A}{B}$ and $\eventTest{G}{y}{Y} \in \Instr{A}{C}$ are \textdef{weakly-compatible} if one does not exclude the other and vice versa.
	\end{definition}
	
	In the particular case of observation-instruments, the compatibility condition takes a simpler form than~\eqref{eqt:doesNotExclude}~\cite{darianoIncompatibilityObservablesChannels2022}.
	
	\begin{definition}[Compatibility of observation-instruments]\label{def:compatibility}
		Consider a causal \ac{OPT} \OPTMath{}. Two observation-instruments \observationEventTest{a}{x}{X}, $\observationEventTest{b}{y}{Y} \in \Obs{A}$ are \textdef{compatible} if there exists a third test $\observationEventTest{c}{(x,y)}{\outcomeSpaceDouble{X}{Y}} \in \Obs{A}$ such that~\cite{darianoIncompatibilityObservablesChannels2022}
		\begin{align*}
			&\myQcircuit{
				&\s{A}\qw&\measureD{\observationEvent{a}{x}}& 
			}  = \quad \sum_{\outcomeIncluded{y}{Y}} \myQcircuit{
				&\s{A}\qw&\measureD{\observationEvent{c}{\left(x,y\right)}}&
			} \quad \forall \outcomeIncluded{x}{X}, \\	
			&\myQcircuit{
				&\s{A}\qw&\measureD{\observationEvent{b}{y}}& 
			}  = \quad \sum_{\outcomeIncluded{x}{X}} 
			\myQcircuit{
				&\s{A}\qw&\measureD{\observationEvent{c}{\left(x,y\right)}}&
			}  \quad \forall \outcomeIncluded{y}{Y}.
		\end{align*}
	\end{definition}
	
	Besides being a special instance of the more general definition involving arbitrary instruments, this is also the \ac{OPT} analogue of the definition of compatibility between positive operator-valued measures~\cite{buschNoInformationDisturbance2009,heinosaariInvitationQuantumIncompatibility2016,heinosaariNofreeinformationPrincipleGeneral2019} in quantum theory.
	
	\begin{remark}[Weak vs strong compatibility]
		The typical approach to compatibility of observation instruments, looking at joint measurability, cannot be straightforwardly generalized to arbitrary instruments. Indeed, for instruments, not only the statistics but also the post-measurement state must be taken into account. Asking for joint implementation of a pair of instruments requires the output system to carry sufficient information to reproduce the outputs of both operations. Of course, when restricted to observation instruments, this is not an issue since the final output is classical and can be copied and post-processed so that the statistics of both measurements can be reproduced. However, for instruments, joint implementability happens to be too restrictive. For this reason, in Ref.~\cite{darianoIncompatibilityObservablesChannels2022} the weaker notion of compatibility in Definition~\ref{def:weak-compatibility} was introduced: Two operations are compatible if there exists a procedure to implement one and later decide to convert it into the other one and vice versa. It has been proved that this condition coincides with joint measurability if applied to observation instruments, while it is strictly weaker for general instruments.
	\end{remark}
	
	\subsection{Complementarity}
	We now formally define the notion of \textdef{complementarity} for generic theories of information processing.
	
	Let us start by generalising the notions of repeatable instruments and verifier states, introducing also the notion of strong verifiers that occur in quantum theory.
	\begin{definition}[Repeatable instruments]\label{def:opt:instr:repeatable}
		Let $\testNoDown{T} = \eventTest{T}{x}{X} \in \InstrA{A}$. We say that $\testNoDown{T}$ is \textdef{repeatable} if 
		\begin{equation*}
			\event{T}{x}\event{T}{x'} = \kronekerDelta{\outcome{x}}{\outcome{x'}} \event{T}{x} \quad \forall \outcome{x},\outcome{x'} \in \outcomeSpace{X}.
		\end{equation*}
	\end{definition}
	
	\begin{definition}[Verifier state]\label{def:opt:state:verifier}
		Consider a generic \ac{OPT} and let $\test{T}{X} \in \Instr{A}{B}$ be an instrument of the theory. We say that a state $\preparationEventNoDown{\rho} \in \StN{AE}$, where \system{E} is a generic system of the theory, is a \textdef{verifier state} for the instrument, denoted by $\preparationEventNoDown{\rho} \in \verSt{\test{T}{X}}$, if there exist $\outcomeIncluded{x}{X}$ such that
		\begin{equation}\label{eqt:opt:state:verifier}
			\myQcircuit{
				&\multiprepareC{1}{\preparationEventNoDown{\rho}}&\s{A}\qw&\gate{\event{T}{x}}&\s{B}\qw&\multimeasureD{1}{\observationUniqueDeterministic_{k}}&
				\\
				&\pureghost{\preparationEventNoDown{\rho}}&\qw&\s{E}\qw&\qw&\ghost{\observationUniqueDeterministic_{k}}&
			} = 1
		\end{equation}
		for all deterministic effects $\observationUniqueDeterministic_{k} \in \EffN{BE}$.
		
		Given the outcome \outcomeIncluded{x}{X} such that \eqref{eqt:opt:state:verifier} is unique, we will also say that $\preparationEventNoDown{\rho}$ is a verifier state for the particular transformation $\event{T}{x}$, denoted by $\preparationEventNoDown{\rho} \in \verSt{\event{T}{x}}$.
	\end{definition}
	
	\begin{remark}\label{rmk:opt:state:verifier:strong}
		It is possible to devise a stronger notion of verifier state. A state $\preparationEventNoDown{\rho} \in \StN{AE}$ is a \textdef{strong verifier} for an instrument $\test{T}{X} \equiv \eventTest{T}{x}{X} \in \Instr{A}{B}$ if there exists \outcomeIncluded{x}{X} such that
		\begin{equation}\label{eqt:opt:state:verifier:strong}
			\myQcircuit{
				&\multiprepareC{1}{\preparationEventNoDown{\rho}}&\s{A}\qw&\gate{\event{T}{x}}&\s{B}\qw&\qw&
				\\
				&\pureghost{\preparationEventNoDown{\rho}}&\qw&\s{E}\qw&\qw&\qw&
			} = \myQcircuit{
				&\multiprepareC{1}{\preparationEventNoDown{\rho}}&\s{B}\qw&\qw&
				\\
				&\pureghost{\preparationEventNoDown{\rho}}&\s{E}\qw&\qw&
			}.
		\end{equation}
		It is immediate to see that if a state satisfies \eqref{eqt:opt:state:verifier:strong}, then it also always satisfies \eqref{eqt:opt:state:verifier}. The converse is in general not true. Consider, for example, the states $\rket{0}$ and $\frac{1}{2}\left( \rket{0} + \rket{1} \right)$ for the transformation
		\begin{align*}
			\myQcircuit{
				&\s{A}\qw&\measureD{\rbra{0} + \rbra{1}}&\prepareC{\rket{0}}&\s{A}\qw&\qw&
			}.
		\end{align*}
		While both states are verifiers for the latter transformation, only the former one is a strong verifier.
	\end{remark}
	
	Exploiting those definitions, we can define \textdef{elementary properties}.
	\begin{definition}[Elementary property]\label{def:elementary-property}
		Let $\testNoDown{P} \equiv \eventTest{P}{x}{X} \in \InstrA{A}$ be a repeatable atomic instrument such that any transformation composing it admits a verifier state, i.e., $\verSt{\event{P}{x}} \neq \emptyset$ for all \outcomeIncluded{x}{X}. We denote by \textdef{elementary property (of system \system{A})} any instrument that is like $\testNoDown{P}$.
	\end{definition}
	
	This definition summarises the two main features that we require from a property: that it can be verified under suitable circumstances (existence of verifier states) and that it holds a notion of objectivity, or better, intersubjectivity (repeatability).
	
	More in general, it could be possible to define a \textdef{property} for a certain system \system{A} as any coarse graining of an elementary property.
	
	Finally, we define what it means for elementary properties to be \textdef{complementary} and the degrees of complementarity that can arise.
	
	\begin{definition}[Complementary elementary properties]\label{def:complementarity}
		Let \testNoDown{P} and \testNoDown{P'} be two elementary properties. They are said to be \textdef{complementary} if there exists $\preparationEventNoDown{\rho} \in \verSt{\testNoDown{P}} \cup \verSt{\testNoDown{P'}}$ such that $\preparationEventNoDown{\rho} \notin \verSt{\testNoDown{P}} \cap \verSt{\testNoDown{P'}}$. In words, two elementary properties are complementary if there exists a verifier state of one that is not a verifier state of the other.
		
		If an \ac{OPT} admits complementary elementary properties, we will say that it has \textdef{complementarity}.
	\end{definition}
	
	We would like to highlight two facts related to the definitions in this section. The first one is that the definition would hold also in the case of non-causal \acp{OPT}. We are not making any assumption on the uniqueness of the deterministic effect. The second one is that we are considering also the presence of an ancillary system \system{E}. This is because we want our definition to take into account the possibility of \acp{OPT} that do not satisfy the property of local discriminability~\cite{darianoHowDeriveHilbertSpace2006,hardyLimitedHolismRealVectorSpace2012,darianoQuantumTheoryFirst2016,darianoClassicalityLocalDiscriminability2020,centenoTwirledWorldsSymmetryinduced2025}. This means that to characterise the state of a multipartite system it is not sufficient to make local measurements, as it is the case in quantum and classical theory. If multiple agents are sharing a system, such as in a Bell-like scenario where Alice and Bob share an entangled pair, then in order to characterise the system it is not sufficient to perform only measurements in the two separated laboratories, but an overall measurement of the composite system would be needed.
	
	Different degrees of complementarity between properties can be specified, depending on how much the definition of one property undermines the knowledge of the other. Entropy provides a natural way to capture this loss of information.
	
	\begin{definition}[Degrees of complementarity]\label{def:opt:compl:degree}
		Let $\testNoDown{P}$ and $\testNoDown{P'} \equiv \eventTest{P'}{x}{X} \in \InstrA{A}$ be two complementary elementary properties and let $\preparationEventNoDown{\nu} \in \verSt{\testNoDown{P}}$ be a verifier state of $\testNoDown{P}$. For every $\observationUniqueDeterministic_{k} \in \EffN{AE}$, define the conditional probabilities
		\begin{align*}
			\myQcircuit{
				&\multiprepareC{1}{\preparationEventNoDown{\nu}}&\s{A}\qw&\gate{\event{P}{x}'}&\s{A}\qw&\multimeasureD{1}{\observationUniqueDeterministic_{k}}&
				\\
				&\pureghost{\preparationEventNoDown{\nu}}&\qw&\s{E}\qw&\qw&\ghost{\observationUniqueDeterministic_{k}}&
			} 
			\;\mathDef\; 
			\probabilityEventNoDown{p}_{\nu}\!\left( \outcome{x} \,\middle|\, \observationUniqueDeterministic_{k} \right)
		\end{align*}
		and the corresponding Shannon entropy
		\begin{align}\label{eqt:compl:deg:shannon}
			\entropy[\nu]{\observationUniqueDeterministic_{k}} 
			\;\mathDef\; - \sum_{\outcomeIncluded{x}{X}}  
			\probabilityEventNoDown{p}_{\nu}\!\left( \outcome{x} \,\middle|\, \observationUniqueDeterministic_{k} \right) 
			\log_2\!\left( \probabilityEventNoDown{p}_{\nu}\!\left( \outcome{x} \,\middle|\, \observationUniqueDeterministic_{k} \right)\right).
		\end{align}
		
		We say that $\testNoDown{P}$ and $\testNoDown{P'}$ are complementary to the following degrees:
		\begin{enumerate}[(i)]
			\item \emph{strong,}\index{Complementarity!Strong}
			\begin{align*}
				\min_{\observationUniqueDeterministic_{k} \in \EffN{AE}} \entropy[\nu]{\observationUniqueDeterministic_{k}} = \log_2\!\left( \cardinality{\outcomeSpace{X}} \right),
			\end{align*}
			where the second property is completely undermined, i.e., every outcome occurs with equal probability, independently of the effect chosen;
			
			\item \emph{mild,}\index{Complementarity!Mild}
			\begin{align*}
				0<\max_{\observationUniqueDeterministic_{k} \in \EffN{AE}} \entropy[\nu]{\observationUniqueDeterministic_{k}} \leq \log_2\!\left( \cardinality{\outcomeSpace{X}} \right),
			\end{align*}
			with $\probabilityEventNoDown{p}_{\nu}\!\left( \outcome{x} \,\middle|\, \observationUniqueDeterministic_{k} \right) \in (0,1)$ for all $\outcomeIncluded{x}{X}$ and for all $\observationUniqueDeterministic_{k}$, where one has some limited information, but every outcome remains possible;
			
			\item \emph{weak,}\index{Complementarity!Weak}
			\begin{align*}
				0<\max_{\observationUniqueDeterministic_{k} \in \EffN{AE}} \entropy[\nu]{\observationUniqueDeterministic_{k}} \leq \log_2\!\left( \cardinality{\outcomeSpace{X}} \right),
			\end{align*}
			with $\probabilityEventNoDown{p}_{\nu}\!\left( \outcome{x} \,\middle|\, \observationUniqueDeterministic_{k} \right) \in [0,1)$ for all $\outcomeIncluded{x}{X}$  and for all $\observationUniqueDeterministic_{k}$.  
			In this case, some outcomes are excluded entirely; hence one has partial knowledge about the second property.
		\end{enumerate}
	\end{definition}
	
	We stress that the distinction between mild and weak complementarity lies in whether all outcomes remain possible or some are strictly forbidden.
	
	Furthermore, using the definition of the Shannon entropy introduced in \autoref{def:opt:compl:degree}, one immediately obtains the following result.
	
	\begin{lemma}
		Let $\testNoDown{P}$ and $\testNoDown{P'} \in \InstrA{A}$ be two elementary properties and let $\preparationEventNoDown{\nu} \in \verSt{\testNoDown{P}}$. Then $\preparationEventNoDown{\nu}$ is a verifier for $\testNoDown{P'}$ if and only if
		\begin{align*}
			\entropy[\preparationEventNoDown{\nu}]{\observationUniqueDeterministic_{k}} = 0 \quad \forall \observationUniqueDeterministic_{k} \in \EffN{AE},
		\end{align*}
		where the entropy $\entropy[\preparationEventNoDown{\nu}]{\observationUniqueDeterministic_{k}}$ is defined in~\eqref{eqt:compl:deg:shannon}.
	\end{lemma}
	
	The relations characterising the different degrees of complementarity presented in Definition~\ref{def:opt:compl:degree} take a particularly simple form in the case of causal \acp{OPT}, where the deterministic effect is unique.
	
	\begin{lemma}[Degrees of complementarity, causal case]\label{lem:causal_opt:compl:degree}
		Consider a causal \ac{OPT} and let $\testNoDown{P}$ and $\testNoDown{P'} \equiv \setNotEnsemble{\event{P}{x}'}_{\outcomeIncluded{x}{X}} \in \InstrA{A}$ be two complementary elementary properties.  Let $\preparationEventNoDown{\nu} \in \verSt{\testNoDown{P}}$, with $\preparationEventNoDown{\nu} \centernot{\in} \verSt{\testNoDown{P'}}$. Then, with respect to the unique deterministic effect $\observationUniqueDeterministic \in \EffN{AE}$, the state $\preparationEventNoDown{\nu}$ generates complementarity between the two properties in one of the following three degrees:
		\begin{enumerate}[(i)]
			\item strong;\index{Complementarity!Strong}
			\begin{align*}
				\myQcircuit{
					&\multiprepareC{1}{\preparationEventNoDown{\nu}}&\s{A}\qw&\gate{\event{P}{x}'}&\s{A}\qw&\multimeasureD{1}{\observationUniqueDeterministic}&
					\\
					&\pureghost{\preparationEventNoDown{\nu}}&\qw&\s{E}\qw&\qw&\ghost{\observationUniqueDeterministic}&
				}= \frac{1}{\cardinality{\outcomeSpace{X}}} \quad \forall \outcomeIncluded{x}{X}
			\end{align*}
			
			\item mild;\index{Complementarity!Mild}
			\begin{align*}
				\myQcircuit{
					&\multiprepareC{1}{\preparationEventNoDown{\nu}}&\s{A}\qw&\gate{\event{P}{x}'}&\s{A}\qw&\multimeasureD{1}{\observationUniqueDeterministic}&
					\\
					&\pureghost{\preparationEventNoDown{\nu}}&\qw&\s{E}\qw&\qw&\ghost{\observationUniqueDeterministic}&
				} \in (0,1) \quad \forall \outcomeIncluded{x}{X}
			\end{align*}
			
			\item weak;\index{Complementarity!Weak}
			\begin{align*}
				\myQcircuit{
					&\multiprepareC{1}{\preparationEventNoDown{\nu}}&\s{A}\qw&\gate{\event{P}{x}'}&\s{A}\qw&\multimeasureD{1}{\observationUniqueDeterministic}&
					\\
					&\pureghost{\preparationEventNoDown{\nu}}&\qw&\s{E}\qw&\qw&\ghost{\observationUniqueDeterministic}&
				} \in [0,1) \quad \forall \outcomeIncluded{x}{X}.
			\end{align*}
		\end{enumerate}
	\end{lemma}
	
	\subsection{Relation between complementarity and incompatibility}
	We now prove the main result of our work, namely, that for generic causal \acp{OPT}, complementarity implies incompatibility.
	We note that, in the next section, this relies on complementarity as a quality, namely, on its on or off nature. This means that the result holds for any degree of complementarity (see Definition~\ref{def:opt:compl:degree} and Lemma~\ref{lem:causal_opt:compl:degree}) and follows as a corollary of the next theorem.
	
	\begin{theorem}\label{thm:opt:verSt:primaryInclusionProperty}
		Let \OPTMath{} be a causal \ac{OPT}. Consider two instruments $\test{T}{X} \equiv \eventTest{T}{x}{X}$ and $\test{G}{Y} \equiv \eventTest{G}{y}{Y} \in \Instr{A}{A}$ such that the former is repeatable and the latter is both atomic and $\event{G}{y} \neq \nullTransformation{A}{A}$ for all \outcomeIncluded{y}{Y}, where $\nullTransformation{A}{A}$ denotes the null transformation for system $\system{A}$. If $\test{T}{X}$ does not exclude $\test{G}{Y}$ $\left(\doNotExclude{\test{T}{X}}{\test{G}{Y}}\right)$, then the set of verifier states of $\test{G}{Y}$ is included within that of $\test{T}{X}$, $\verSt{\test{G}{Y}} \subseteq \verSt{\test{T}{X}}$.
	\end{theorem}
	
	\begin{proof}
		Let us start by explicitly stating the tests that characterise the relationship $\doNotExclude{\test{T}{X}}{\test{G}{Y}}$:
		\begin{equation*}
			\myQcircuit{
				&\s{A}\qw&\gate{\event{T}{x}}&\s{A}\qw&\qw&
			} = \sum_{\outcomeIncludedConditioned{z}{Z}{x}} \myQcircuit{
				&\s{A}\qw&\multigate{1}{\event{C}{z}}&\s{A}\qw&\qw&\qw&
				\\
				&\pureghost{}&\pureghost{\event{C}{z}}&\s{E}\qw&\measureD{\observationUniqueDeterministic}&
			} \quad \forall \outcomeIncluded{x}{X}
		\end{equation*}
		and 
		\begin{equation*}
			\myQcircuit{
				&\s{A}\qw&\gate{\event{G}{y}}&\s{A}\qw&\qw&
			} = \sum_{\outcomeIncluded{z}{Z}} \myQcircuit{
				&\s{A}\qw&\multigate{1}{\event{C}{z}}&\s{A}\qw&\multigate{1}{\conditionedEvent{P}{y}{z}}&\s{A}\qw&\qw&
				\\
				&\pureghost{}&\pureghost{\event{C}{z}}&\s{E}\qw&\ghost{\conditionedEvent{P}{y}{z}}&
			} \quad \forall \outcomeIncluded{y}{Y}.
		\end{equation*}
		Exploiting the repeatability property of \test{T}{X}, one can prove by direct calculation that
		\begin{align}\label{eqt:opt:complem:strongIncop:1}
			\myQcircuit{
				&\s{A}\qw&\multigate{2}{\event{C}{z}}&\s{A}\qw&\multigate{1}{\event{C}{z'}}&\s{A}\qw&\qw&
				\\
				&\pureghost{}&\pureghost{\event{C}{z}}&\pureghost{}&\pureghost{\event{C}{z'}}&\s{E'}\qw&\qw&
				\\
				&\pureghost{}&\pureghost{\event{C}{z}}&\qw&\qw&\s{E}\qw&\qw&
			} = \nullTransformation{A}{A},
		\end{align}
		whenever $\outcomeIncludedConditioned{z}{Z}{x}$ and $\outcomeIncludedConditioned{z'}{Z}{x'}$ with $\outcome{x} \neq \outcome{x'}$. In particular, the latter property follows from the fact that $\event{T}{x}\event{T}{x'} = \event{T}{x'}\event{T}{x} = \nullTransformation{A}{A}$ whenever $\outcome{x} \neq \outcome{x'}$.
		
		Let us now define
		\begin{equation*}
			\myQcircuit{
				&\s{A}\qw&\gate{\event{D}{y,x}}&\s{A}\qw&\qw&
			} \mathDef \sum_{\outcomeIncludedConditioned{z}{Z}{x}} \myQcircuit{
				&\s{A}\qw&\multigate{1}{\event{C}{z}}&\s{A}\qw&\multigate{1}{\conditionedEvent{P}{y}{z}}&\s{A}\qw&\qw&
				\\
				&\pureghost{}&\pureghost{\event{C}{z}}&\s{E}\qw&\ghost{\conditionedEvent{P}{y}{z}}&
			}.
		\end{equation*}
		Exploiting \eqref{eqt:opt:complem:strongIncop:1}, it is then immediate to conclude that 
		\begin{align}\label{eqt:opt:complem:strongIncop:2}
			\myQcircuit{
				&\s{A}\qw&\gate{\event{T}{x}}&\s{A}\qw&\gate{\event{D}{y,x'}}&\s{A}\qw&\qw&
			} = \quad\!\! \nullTransformation{A}{A},
		\end{align}
		whenever $\outcome{x}\neq\outcome{x'}$ for all \outcomeIncluded{y}{Y} and then 
		\begin{align}\label{eqt:opt:complem:strongIncop:3}
			\myQcircuit{
				&\s{A}\qw&\gate{\event{T}{x}}&\s{A}\qw&\gate{\event{G}{y}}&\s{A}\qw&\qw&
			} = \quad\!\! \myQcircuit{
				&\s{A}\qw&\gate{\event{T}{x}}&\s{A}\qw&\gate{\event{D}{y,x}}&\s{A}\qw&\qw&
			}
		\end{align}
		since $\event{G}{y} = \sum_{\outcomeIncluded{x}{X}} \event{D}{y,x}$. Due to the fact that $\test{G}{Y}$ is atomic, it must then be
		\begin{align*}
			\myQcircuit{
				&\s{A}\qw&\gate{\event{D}{y,x}}&\s{A}\qw&\qw&
			} = \quad\!\! \lambda_{\outcome{y,x}} \myQcircuit{
				&\s{A}\qw&\gate{\event{G}{y}}&\s{A}\qw&\qw&
			}
		\end{align*}
		for some $\lambda_{\outcome{y,x}} \in [0,1]$. Consequently, due to~\eqref{eqt:opt:complem:strongIncop:3}, one gets
		\begin{align}\label{eqt:opt:complem:strongIncop:4}
			\myQcircuit{
				&\s{A}\qw&\gate{\event{T}{x}}&\s{A}\qw&\gate{\event{G}{y}}&\s{A}\qw&\qw&
			} = \quad\!\! \lambda_{\outcome{y,x}} \myQcircuit{
				&\s{A}\qw&\gate{\event{T}{x}}&\s{A}\qw&\gate{\event{G}{y}}&\s{A}\qw&\qw&
			}.
		\end{align}
		Equation~\eqref{eqt:opt:complem:strongIncop:4} implies that it must be $\lambda_{\outcome{y,x}}\in \left\{0,1\right\}$. Furthermore, for any \outcomeIncluded{y}{Y} there exists a unique \outcomeIncluded{x}{X} such that $\lambda_{\outcome{y,x}} = 1$. To prove both the existence and the uniqueness, we will proceed by contradiction, starting from the former. Suppose that for some $\outcomeIncluded{y}{Y}$ the constant is always zero, namely, $\lambda_{\outcome{y,x}} = 0$ for all \outcomeIncluded{x}{X}. Then
		\begin{align*}
			\myQcircuit{
				&\s{A}\qw&\gate{\event{G}{y}}&\s{A}\qw&\qw&
			} &= \sum_{\outcomeIncluded{x}{X}} \myQcircuit{
				&\s{A}\qw&\gate{\event{D}{y,x}}&\s{A}\qw&\qw&
			}\\[10pt]
			&= \sum_{\outcomeIncluded{x}{X}} \lambda_{\outcome{y,x}} \myQcircuit{
				&\s{A}\qw&\gate{\event{G}{y}}&\s{A}\qw&\qw&
			}
			= \nullTransformation{A}{A},
		\end{align*}
		which is absurd, since by hypothesis $\event{G}{y} \neq \nullTransformation{A}{A}$ for all \outcomeIncluded{y}{Y}. To prove the uniqueness it is sufficient to observe that if for a given \outcomeIncluded{y}{Y} there exists a set of outcomes $\outcomeIncluded{x}{X'}$ such that $\lambda_{\outcome{y,x}} = 1$ for all $\outcomeIncluded{x}{X'}$, then
		\begin{align*}
			\myQcircuit{
				&\s{A}\qw&\gate{\event{G}{y}}&\s{A}\qw&\qw&
			} &= \sum_{\outcomeIncluded{x}{X}} \myQcircuit{
				&\s{A}\qw&\gate{\event{D}{y,x}}&\s{A}\qw&\qw&
			}\\[10pt]
			&= \sum_{\outcomeIncluded{x}{X}} \lambda_{\outcome{y,x}} \myQcircuit{
				&\s{A}\qw&\gate{\event{G}{y}}&\s{A}\qw&\qw&
			}\\[10pt]
			&= \sum_{\outcomeIncluded{x}{X'}} \lambda_{\outcome{y,x}} \myQcircuit{
				&\s{A}\qw&\gate{\event{G}{y}}&\s{A}\qw&\qw&
			}\\[10pt]
			&= \sum_{\outcomeIncluded{x}{X'}} \myQcircuit{
				&\s{A}\qw&\gate{\event{G}{y}}&\s{A}\qw&\qw&
			}\\[10pt]
			&= \cardinality{\outcomeSpace{X'}} \myQcircuit{
				&\s{A}\qw&\gate{\event{G}{y}}&\s{A}\qw&\qw&
			},
		\end{align*}
		which is possible only if $\cardinality{\outcomeSpace{X}'} = 1$.		
		To conclude the proof let us then consider a verifier state $\preparationEventNoDown{\rho}$ for $\test{G}{Y}$. By definition we know that there exists $\outcomeIncluded{y}{Y}$ such that $\preparationEventNoDown{\rho} \in \verSt{\event{G}{y}}$, or in the formula
		\begin{align*}
			\myQcircuit{
				&\multiprepareC{1}{\preparationEventNoDown{\rho}}&\s{A}\qw&\gate{\event{G}{y}}&\s{A}\qw&\multimeasureD{1}{\observationUniqueDeterministic}&
				\\
				&\pureghost{\preparationEventNoDown{\rho}}&\qw&\s{E'}\qw&\qw&\ghost{\observationUniqueDeterministic}&
			} = 1,
		\end{align*}
		where $\system{E'}$ is an appropriate ancillary system. However, the latter relation can also be rewritten as
		\begin{align*}
			\sum_{\outcomeIncludedConditioned{z}{Z}{x}} \myQcircuit{
				&\multiprepareC{2}{\preparationEventNoDown{\rho}}&\s{A}\qw&\multigate{1}{\event{C}{z}}&\s{A}\qw&\multigate{1}{\conditionedEvent{P}{y}{z}}&\s{A}\qw&\multimeasureD{2}{\observationUniqueDeterministic}&
				\\
				&\pureghost{\preparationEventNoDown{\rho}}&\pureghost{}&\pureghost{\event{C}{z}}&\s{E}\qw&\ghost{\conditionedEvent{P}{y}{z}}&
				\\
				&\pureghost{\preparationEventNoDown{\rho}}&\qw&\qw&\s{E'}\qw&\qw&\qw&\ghost{\observationUniqueDeterministic}&
			} = 1,
		\end{align*}
		where we exploited the fact that for all $\outcomeIncluded{y}{Y}$ there exists a unique $\outcomeIncluded{x}{X}$ such that
		\begin{align*}
			\myQcircuit{
				&\s{A}\qw&\gate{\event{G}{y}}&\s{A}\qw&\qw&
			} = \quad\!\! \myQcircuit{
				&\s{A}\qw&\gate{\event{D}{y,x}}&\s{A}\qw&\qw&
			}.
		\end{align*}
		Now, given that in general $\conditionedEvent{P}{y}{z}$ is not a deterministic transformation, it must also hold that
		\begin{align*}
			&\sum_{\outcomeIncludedConditioned{z}{Z}{x}} \myQcircuit{
				&\multiprepareC{2}{\preparationEventNoDown{\rho}}&\s{A}\qw&\multigate{1}{\event{C}{z}}&\s{A}\qw&\multigate{1}{\conditionedEvent{P}{y}{z}}&\s{A}\qw&\multimeasureD{2}{\observationUniqueDeterministic}&
				\\
				&\pureghost{\preparationEventNoDown{\rho}}&\pureghost{}&\pureghost{\event{C}{z}}&\s{E}\qw&\ghost{\conditionedEvent{P}{y}{z}}&
				\\
				&\pureghost{\preparationEventNoDown{\rho}}&\qw&\qw&\s{E'}\qw&\qw&\qw&\ghost{\observationUniqueDeterministic}&
			}\\[10pt]
			&\leq \sum_{\outcomeIncludedConditioned{z}{Z}{x}} \myQcircuit{
				&\multiprepareC{2}{\preparationEventNoDown{\rho}}&\s{A}\qw&\multigate{1}{\event{C}{z}}&\s{A}\qw&\multimeasureD{2}{\observationUniqueDeterministic}&
				\\
				&\pureghost{\preparationEventNoDown{\rho}}&\pureghost{}&\pureghost{\event{C}{z}}&\s{E}\qw&\ghost{\observationUniqueDeterministic}&
				\\
				&\pureghost{\preparationEventNoDown{\rho}}&\qw&\s{E'}\qw&\qw&\ghost{\observationUniqueDeterministic}&
			}\\[10pt]
			&= \sum_{\outcomeIncludedConditioned{z}{Z}{x}} \myQcircuit{
				&\multiprepareC{2}{\preparationEventNoDown{\rho}}&\s{A}\qw&\multigate{1}{\event{C}{z}}&\s{A}\qw&\measureD{\observationUniqueDeterministic}&
				\\
				&\pureghost{\preparationEventNoDown{\rho}}&\pureghost{}&\pureghost{\event{C}{z}}&\s{E}\qw&\measureD{\observationUniqueDeterministic}&
				\\
				&\pureghost{\preparationEventNoDown{\rho}}&\qw&\s{E'}\qw&\qw&\measureD{\observationUniqueDeterministic}&
			}\\[10pt]
			&= \myQcircuit{
				&\multiprepareC{1}{\preparationEventNoDown{\rho}}&\s{A}\qw&\gate{\event{T}{x}}&\s{A}\qw&\multimeasureD{1}{\observationUniqueDeterministic}&
				\\
				&\pureghost{\preparationEventNoDown{\rho}}&\qw&\s{E'}\qw&\qw&\ghost{\observationUniqueDeterministic}&
			} \leq 1,
		\end{align*}
		where in the last step we used that, in general, $\event{T}{x}$ is not a deterministic transformation. Putting everything together, one obtains that it must be
		\begin{align*}
			\myQcircuit{
				&\multiprepareC{1}{\preparationEventNoDown{\rho}}&\s{A}\qw&\gate{\event{T}{x}}&\s{A}\qw&\multimeasureD{1}{\observationUniqueDeterministic}&
				\\
				&\pureghost{\preparationEventNoDown{\rho}}&\qw&\s{E'}\qw&\qw&\ghost{\observationUniqueDeterministic}&
			} = 1.
		\end{align*}
		Hence, $\preparationEventNoDown{\rho} \in \verSt{\test{T}{X}}$, or more precisely $\preparationEventNoDown{\rho} \in \verSt{\event{T}{x}}$. Given the arbitrariness of $\preparationEventNoDown{\rho} \in \verSt{\test{G}{Y}}$, one obtains that $\verSt{\test{G}{Y}} \subseteq  \verSt{\test{T}{X}}$.
	\end{proof}

	\begin{corollary}\label{cor:complementarity-imnplies-incompatibility}
		In causal \acp{OPT} complementary elementary properties are strongly incompatible or, equivalently, weakly compatible elementary properties are non-complementary.
	\end{corollary}
	
	\begin{proof}
		Considering two elementary properties $\test{P}{X}$ and $\test{P}{Y}' \in \InstrA{A}$. If one supposes that they are weakly compatible, then \autoref{thm:opt:verSt:primaryInclusionProperty} guarantees that $\verSt{\test{P}{X}} \equiv \verSt{\testNoDown{P}'_{\outcomeSpace{Y}}}$, i.e.~that these two elementary properties are non-complementary.
	\end{proof}
	
	\section{Conclusion}
	\label{sec:conclusion}
	
	Summarizing, we have shown that if two elementary properties in quantum theory are non-complementary, they must be the same property. This trivially implies that the two properties are compatible, also according to the usual textbook definition of compatibility for observables, i.e.~commutativity of their projection-valued measures~\cite{buschComparingDegreesIncompatibility2013,heinosaariSimultaneousMeasurementTwo2016,heinosaariInvitationQuantumIncompatibility2016,buscemiUnifyingDifferentNotions2023}.
	
	In this work we formalized in an information-theoretic setting the notion of complementarity originally introduced by Bohr. Building on this informational definition, grounded in the concept of properties, understood as ensembles of repeatable atomic operations, we provided a characterization of non-complementary properties in quantum theory. Therefore, we proved the equivalence between complementarity and incompatibility of quantum operations. The study of the relationship between the two dates back to the origins of quantum mechanics, in the debate between Bohr and Heisenberg. We have generalized the present notions from quantum mechanics to an arbitrary information theory and show that the implication that complementarity implies incompatibility always holds for causal theories. The opposite implication is instead rooted in the characterization of verifiers as invariant states under the action of an elementary property.
	
	As a future perspective it would be relevant to understand complementarity as a resource and better comprehend how the measure of complementarity proposed here could be related to the measures of incompatibility known in the literature~\cite{buscemiUnifyingDifferentNotions2023}.
	
	It would also be interesting to investigate the connection of complementarity with classicality (in the operational framework considered by us, captured by the notion of simpliciality~\cite{darianoClassicalityLocalDiscriminability2020}) in order to assess whether complementarity could provide a suitable notion to distinguish the quantum from the classical world.
	
	\acknowledgements
	D.R.~worked on this project while visiting Perimeter Institute for Theoretical Physics and would like to thank the quantum foundations group for their hospitality during his visit. This research was supported in part by Perimeter Institute for Theoretical Physics. Research at Perimeter Institute is supported by the Government of Canada through the Department of Innovation, Science, and Economic Development, and by the Province of Ontario through the Ministry of Colleges and Universities. D.R.~also acknowledges financial support from the Foundation Blanceflor Boncompagni Ludovisi, née Bildt. A.T.~acknowledges financial support from Elvia and Federico Faggin Foundation (Silicon Valley Community Foundation Project No.~2020-214365). P.P.~acknowledges financial support from European Union - Next Generation EU through the PNNR MUR Project No.~PE0000023-NQSTI. All the authors would like to thank Marco Erba for insightful discussions on the notion of property in the framework of operational probabilistic theories.
	
	\bibliography{complementarity}

\begin{thebibliography}{10}

\bibitem{bohrQuantumPostulateRecent1928}
N.~Bohr.
\newblock ``The {{Quantum Postulate}} and the {{Recent Development}} of
  {{Atomic Theory}}''.
\newblock \href{https://dx.doi.org/10.1038/121580a0}{Nature {\bf 121},
  580--590}~(1928).

\bibitem{bohrCanQuantumMechanicalDescription1935}
N.~Bohr.
\newblock ``Can {{Quantum-Mechanical Description}} of {{Physical Reality}} be
  {{Considered Complete}}?''.
\newblock \href{https://dx.doi.org/10.1103/PhysRev.48.696}{Physical Review {\bf
  48}, 696--702}~(1935).

\bibitem{pauliGeneralPrinciplesQuantum1980}
W.~Pauli.
\newblock ``General {{Principles}} of {{Quantum Mechanics}}''.
\newblock \href{https://dx.doi.org/10.1007/978-3-642-61840-6}{Springer Berlin
  Heidelberg}. Berlin, Heidelberg~(1980).

\bibitem{Saunders:2005aa}
S.~Saunders.
\newblock ``Complementarity and scientific rationality''.
\newblock \href{https://dx.doi.org/10.1007/s10701-004-1982-x}{Foundations of
  Physics {\bf 35}, 417--447}~(2005).

\bibitem{englertFringeVisibilityWhichWay1996}
B.-G. Englert.
\newblock ``Fringe visibility and which-way information: An inequality''.
\newblock \href{https://dx.doi.org/10.1103/PhysRevLett.77.2154}{Physical Review
  Letters {\bf 77}, 2154--2157}~(1996).

\bibitem{scullyQuantumOpticalTests1991}
M.~O. Scully, B.-G. Englert, and H.~Walther.
\newblock ``Quantum optical tests of complementarity''.
\newblock \href{https://dx.doi.org/10.1038/351111a0}{Nature {\bf 351},
  111--116}~(1991).

\bibitem{jaegerTwoInterferometricComplementarities1995}
G.~Jaeger, A.~Shimony, and L.~Vaidman.
\newblock ``Two interferometric complementarities''.
\newblock \href{https://dx.doi.org/10.1103/PhysRevA.51.54}{Physical Review A
  {\bf 51}, 54--67}~(1995).

\bibitem{durrOriginQuantummechanicalComplementarity1998}
S.~D{\"u}rr, T.~Nonn, and G.~Rempe.
\newblock ``Origin of quantum-mechanical complementarity probed by a
  `which-way' experiment in an atom interferometer''.
\newblock \href{https://dx.doi.org/10.1038/25653}{Nature {\bf 395},
  33--37}~(1998).

\bibitem{bimonteInterferometricDualityMultibeam2003}
G.~Bimonte and R.~Musto.
\newblock ``On interferometric duality in multibeam experiments''.
\newblock \href{https://dx.doi.org/10.1088/0305-4470/36/45/009}{Journal of
  Physics A: Mathematical and General {\bf 36}, 11481}~(2003).

\bibitem{martinez-linaresQualityWhichwayDetector2004}
J.~{Martinez-Linares} and D.~A. Harmin.
\newblock ````{{Quality}}'' of a which-way detector''.
\newblock \href{https://dx.doi.org/10.1103/PhysRevA.69.062109}{Physical Review
  A {\bf 69}, 062109}~(2004).

\bibitem{schillingPhasedependentWhichwayInformation2012}
U.~Schilling and J.~{von Zanthier}.
\newblock ``Phase-dependent which-way information''.
\newblock \href{https://dx.doi.org/10.1016/j.physleta.2012.10.019}{Physics
  Letters A {\bf 376}, 3479--3484}~(2012).

\bibitem{dezelaCommentFringeVisibility2013}
F.~{De Zela}.
\newblock ``Comment on ``{F}ringe {V}isibility and {W}hich-{W}ay {I}nformation:
  {A}n {I}nequality''''~(2013).
\newblock  \href{http://arxiv.org/abs/1305.0861}{arXiv:1305.0861}.

\bibitem{prabhutejQuantumWhichwayInformation2014}
J.~Prabhu~Tej, A.~R. Usha~Devi, H.~S. Karthik, {Sudha}, and A.~K. Rajagopal.
\newblock ``Quantum which-way information and fringe visibility when the
  detector is entangled with an ancilla''.
\newblock \href{https://dx.doi.org/10.1103/PhysRevA.89.062116}{Physical Review
  A {\bf 89}, 062116}~(2014).

\bibitem{krausComplementaryObservablesUncertainty1987}
K.~Kraus.
\newblock ``Complementary observables and uncertainty relations''.
\newblock \href{https://dx.doi.org/10.1103/PhysRevD.35.3070}{Physical Review D
  {\bf 35}, 3070--3075}~(1987).

\bibitem{buschOperationalQuantumPhysics1995}
P.~Busch, M.~Grabowski, and P.~J. Lahti.
\newblock ``Operational {{Quantum Physics}}''.
\newblock \href{https://dx.doi.org/10.1007/978-3-540-49239-9}{Volume~31 of
  Lecture {{Notes}} in {{Physics Monographs}}}.
\newblock Springer Berlin Heidelberg. Berlin, Heidelberg~(1995).

\bibitem{buschComplementarityQuantumObservables1995}
P.~Busch and P.~J. Lahti.
\newblock ``The complementarity of quantum observables: {{Theory}} and
  experiments''.
\newblock \href{https://dx.doi.org/10.1007/BF02743814}{La Rivista del Nuovo
  Cimento {\bf 18}, 1--27}~(1995).

\bibitem{buschHeisenbergsUncertaintyPrinciple2007}
P.~Busch, T.~Heinonen, and P.~Lahti.
\newblock ``Heisenberg's uncertainty principle''.
\newblock \href{https://dx.doi.org/10.1016/j.physrep.2007.05.006}{Physics
  Reports {\bf 452}, 155--176}~(2007).

\bibitem{petzComplementarityQuantumSystems2007}
D.~Petz.
\newblock ``Complementarity in quantum systems''.
\newblock \href{https://dx.doi.org/10.1016/S0034-4877(07)00010-9}{Reports on
  Mathematical Physics {\bf 59}, 209--224}~(2007).

\bibitem{plotnitskyWhatComplementarityNiels2014}
A.~Plotnitsky.
\newblock ``What is complementarity?: {{Niels Bohr}} and the architecture of
  quantum theory''.
\newblock \href{https://dx.doi.org/10.1088/0031-8949/2014/T163/014002}{Physica
  Scripta {\bf T163}, 014002}~(2014).

\bibitem{sahaOperationalFoundationsComplementarity2020}
D.~Saha, M.~Oszmaniec, L.~Czekaj, M.~Horodecki, and R.~Horodecki.
\newblock ``Operational foundations for complementarity and uncertainty
  relations''.
\newblock \href{https://dx.doi.org/10.1103/PhysRevA.101.052104}{Physical Review
  A {\bf 101}, 052104}~(2020).

\bibitem{hsiehQuantumComplementarityNovel2023}
C.-Y. Hsieh, R.~Uola, and P.~Skrzypczyk.
\newblock ``Quantum complementarity: {{A}} novel resource for unambiguous
  exclusion and encryption''~(2023).
\newblock  \href{http://arxiv.org/abs/2309.11968}{arXiv:2309.11968}.

\bibitem{serinoComplementaritybasedComplementarity2024}
L.~Serino, G.~Chesi, B.~Brecht, L.~Maccone, C.~Macchiavello, and C.~Silberhorn.
\newblock ``Complementarity-based complementarity: The choice of mutually
  unbiased observables shapes quantum uncertainty relations''.
\newblock \href{https://dx.doi.org/10.1103/v24q-sl6n}{Phys. Rev. Res. {\bf 7},
  033152}~(2025).

\bibitem{buschComparingDegreesIncompatibility2013}
P.~Busch, T.~Heinosaari, J.~Schultz, and N.~Stevens.
\newblock ``Comparing the degrees of incompatibility inherent in probabilistic
  physical theories''.
\newblock \href{https://dx.doi.org/10.1209/0295-5075/103/10002}{Europhysics
  Letters {\bf 103}, 10002}~(2013).

\bibitem{heinosaariSimultaneousMeasurementTwo2016}
T.~Heinosaari.
\newblock ``Simultaneous measurement of two quantum observables:
  {{Compatibility}}, broadcasting, and in-between''.
\newblock \href{https://dx.doi.org/10.1103/PhysRevA.93.042118}{Physical Review
  A {\bf 93}, 042118}~(2016).

\bibitem{heinosaariInvitationQuantumIncompatibility2016}
T.~Heinosaari, T.~Miyadera, and M.~Ziman.
\newblock ``An invitation to quantum incompatibility''.
\newblock \href{https://dx.doi.org/10.1088/1751-8113/49/12/123001}{Journal of
  Physics A: Mathematical and Theoretical {\bf 49}, 123001}~(2016).

\bibitem{darianoIncompatibilityObservablesChannels2022}
G.~M. D'Ariano, P.~Perinotti, and A.~Tosini.
\newblock ``Incompatibility of observables, channels and instruments in
  information theories''.
\newblock \href{https://dx.doi.org/10.1088/1751-8121/ac88a7}{Journal of Physics
  A: Mathematical and Theoretical {\bf 55}, 394006}~(2022).

\bibitem{buscemiUnifyingDifferentNotions2023}
F.~Buscemi, K.~Kobayashi, S.~Minagawa, P.~Perinotti, and A.~Tosini.
\newblock ``Unifying different notions of quantum incompatibility into a strict
  hierarchy of resource theories of communication''.
\newblock \href{https://dx.doi.org/10.22331/q-2023-06-07-1035}{Quantum {\bf 7},
  1035}~(2023).

\bibitem{hardyDisentanglingNonlocalityTeleportation1999}
L.~Hardy.
\newblock ``Disentangling {{Nonlocality}} and {{Teleportation}}''~(1999).
\newblock
  \href{http://arxiv.org/abs/quant-ph/9906123}{arXiv:quant-ph/9906123}.

\bibitem{barnumCloningBroadcastingGeneric2006}
H.~Barnum, J.~Barrett, M.~Leifer, and A.~Wilce.
\newblock ``Cloning and {{Broadcasting}} in {{Generic Probabilistic
  Theories}}''~(2006).
\newblock
  \href{http://arxiv.org/abs/quant-ph/0611295}{arXiv:quant-ph/0611295}.

\bibitem{barrettInformationProcessingGeneralized2007}
J.~Barrett.
\newblock ``Information processing in generalized probabilistic theories''.
\newblock \href{https://dx.doi.org/10.1103/PhysRevA.75.032304}{Physical Review
  A {\bf 75}, 032304}~(2007).

\bibitem{spekkensEvidenceEpistemicView2007}
R.~W. Spekkens.
\newblock ``Evidence for the epistemic view of quantum states: {{A}} toy
  theory''.
\newblock \href{https://dx.doi.org/10.1103/PhysRevA.75.032110}{Physical Review
  A {\bf 75}, 032110}~(2007).

\bibitem{chiribellaProbabilisticTheoriesPurification2010}
G.~Chiribella, G.~M. D'Ariano, and P.~Perinotti.
\newblock ``Probabilistic theories with purification''.
\newblock \href{https://dx.doi.org/10.1103/PhysRevA.81.062348}{Physical Review
  A {\bf 81}, 062348}~(2010).

\bibitem{janottaGeneralizedProbabilisticTheories2013}
P.~Janotta and R.~Lal.
\newblock ``Generalized probabilistic theories without the no-restriction
  hypothesis''.
\newblock \href{https://dx.doi.org/10.1103/PhysRevA.87.052131}{Physical Review
  A {\bf 87}, 052131}~(2013).

\bibitem{darianoQuantumTheoryFirst2016}
G.~M. D'Ariano, G.~Chiribella, and P.~Perinotti.
\newblock ``Quantum {{Theory}} from {{First Principles}}: {{An Informational
  Approach}}''.
\newblock \href{https://dx.doi.org/10.1017/9781107338340}{Cambridge University
  Press}. ~(2016).
\newblock 1 edition.

\bibitem{plavalaGeneralProbabilisticTheories2021}
M.~Pl{\'a}vala.
\newblock ``General probabilistic theories: {{An}} introduction''.
\newblock \href{https://dx.doi.org/10.1016/j.physrep.2023.09.001}{Physics
  Reports {\bf 1033}, 1--64}~(2023).

\bibitem{heisenbergUeberAnschaulichenInhalt1927}
W.~Heisenberg.
\newblock ``{{\"U}ber den anschaulichen Inhalt der quantentheoretischen
  Kinematik und Mechanik}''.
\newblock \href{https://dx.doi.org/10.1007/BF01397280}{Zeitschrift f{\"u}r
  Physik {\bf 43}, 172--198}~(1927).

\bibitem{chiribellaInformationalDerivationQuantum2011}
G.~Chiribella, G.~M. D'Ariano, and P.~Perinotti.
\newblock ``Informational derivation of quantum theory''.
\newblock \href{https://dx.doi.org/10.1103/PhysRevA.84.012311}{Physical Review
  A {\bf 84}, 012311}~(2011).

\bibitem{chiribellaQuantumPrinciples2016}
G.~Chiribella, G.~M. D'Ariano, and P.~Perinotti.
\newblock ``Quantum from {{Principles}}''.
\newblock In G.~Chiribella and R.~W. Spekkens, editors, Quantum {{Theory}}:
  {{Informational Foundations}} and {{Foils}}.
\newblock \href{https://dx.doi.org/10.1007/978-94-017-7303-4_6}{Volume 181,
  pages 171--221}.
\newblock Springer Netherlands, Dordrecht~(2016).

\bibitem{darianoClassicalityLocalDiscriminability2020}
G.~M. D'Ariano, M.~Erba, and P.~Perinotti.
\newblock ``Classicality without local discriminability: {{Decoupling}}
  entanglement and complementarity''.
\newblock \href{https://dx.doi.org/10.1103/PhysRevA.102.052216}{Physical Review
  A {\bf 102}, 052216}~(2020).

\bibitem{perinottiCellularAutomataOperational2020}
P.~Perinotti.
\newblock ``Cellular automata in operational probabilistic theories''.
\newblock \href{https://dx.doi.org/10.22331/q-2020-07-09-294}{Quantum {\bf 4},
  294}~(2020).

\bibitem{rolinoMinimalOperationalTheories2025}
D.~Rolino, M.~Erba, A.~Tosini, and P.~Perinotti.
\newblock ``Minimal operational theories: Classical theories with quantum
  features''.
\newblock \href{https://dx.doi.org/10.1088/1367-2630/ada850}{New Journal of
  Physics {\bf 27}, 023004}~(2025).

\bibitem{leppajarviPostprocessingQuantumInstruments2021}
L.~Lepp{\"a}j{\"a}rvi and M.~Sedl{\'a}k.
\newblock ``Postprocessing of quantum instruments''.
\newblock \href{https://dx.doi.org/10.1103/PhysRevA.103.022615}{Physical Review
  A {\bf 103}, 022615}~(2021).

\bibitem{erbaMeasurementIncompatibilityStrictly2024}
M.~Erba, P.~Perinotti, D.~Rolino, and A.~Tosini.
\newblock ``Measurement incompatibility is strictly stronger than
  disturbance''.
\newblock \href{https://dx.doi.org/10.1103/PhysRevA.109.022239}{Physical Review
  A {\bf 109}, 022239}~(2024).

\bibitem{buschNoInformationDisturbance2009}
P.~Busch.
\newblock ````{{No Information Without Disturbance}}'': {{Quantum Limitations}}
  of {{Measurement}}''.
\newblock In Quantum {{Reality}}, {{Relativistic Causality}}, and {{Closing}}
  the {{Epistemic Circle}}: {{Essays}} in {{Honour}} of {{Abner Shimony}}.
\newblock \href{https://dx.doi.org/10.1007/978-1-4020-9107-0_13}{Pages
  229--256}.
\newblock The {{Western Ontario Series}} in {{Philosophy}} of {{Science}}.
  Springer Netherlands, Dordrecht~(2009).

\bibitem{heinosaariNofreeinformationPrincipleGeneral2019}
T.~Heinosaari, L.~Lepp{\"a}j{\"a}rvi, and M.~Pl{\'a}vala.
\newblock ``No-free-information principle in general probabilistic theories''.
\newblock \href{https://dx.doi.org/10.22331/q-2019-07-08-157}{Quantum {\bf 3},
  157}~(2019).

\bibitem{darianoHowDeriveHilbertSpace2006}
G.~M. D'Ariano.
\newblock ``How to {{Derive}} the {{Hilbert-Space Formulation}} of {{Quantum
  Mechanics From Purely Operational Axioms}}''.
\newblock In {{AIP Conference Proceedings}}.
\newblock \href{https://dx.doi.org/10.1063/1.2219356}{Volume 844, pages
  101--128}.
\newblock Trieste (Italy) and Losinj (Croatia)~(2006). AIP.

\bibitem{hardyLimitedHolismRealVectorSpace2012}
L.~Hardy and W.~K. Wootters.
\newblock ``Limited {{Holism}} and {{Real-Vector-Space Quantum Theory}}''.
\newblock \href{https://dx.doi.org/10.1007/s10701-011-9616-6}{Foundations of
  Physics {\bf 42}, 454--473}~(2012).

\bibitem{centenoTwirledWorldsSymmetryinduced2025}
D.~Centeno, M.~Erba, T.~D. Galley, D.~Schmid, J.~H. Selby, R.~W. Spekkens,
  S.~Soltani, J.~Surace, A.~Wilce, and Y.~Y{\ifmmode \bar{\imath}\else
  \={\i}\fi{}}ng.
\newblock ``Symmetry-induced failures of tomographic locality: Constructing
  foil theories by twirling''.
\newblock \href{https://dx.doi.org/10.1103/hpmv-15sf}{Phys. Rev. A {\bf 112},
  L030202}~(2025).

\end{thebibliography}
	\bibliographystyle{quantum}
	
\end{document}